\documentclass[11pt, a4paper]{article}

\usepackage{a4wide, amsfonts,enumerate,amsmath,amsfonts,xspace,array,delarray,theorem,amssymb,epsfig}
\usepackage{graphicx,color}

\theoremstyle{break}\theorembodyfont{\it}
\theoremheaderfont{\normalfont\bfseries}

\newtheorem{theo}{Theorem}

\newtheorem{lem}[theo]{Lemma}

\newtheorem{prop}[theo]{Proposition}

\newtheorem{coro}[theo]{Corollary}

\newtheorem{rem}[theo]{Remark}

\newenvironment{proof}{\noindent{\bf Proof: }}
                {\leavevmode\unskip\nobreak\hskip2em plus1fill
                $\scriptstyle\square$\vskip\theorempostskipamount\par}

\def\exp{\mathrm {exp}}
\def\dis{\displaystyle}

\let\lt=<
\let\gt=>

\def\R{{\mathbb R}}

\let\phi=\varphi
\let\eps=\varepsilon

\def\be{\begin{equation}}
\def\ee{\end{equation}}
\def\beq{\begin{eqnarray*}}
\def\eeq{\end{eqnarray*}}
\def\e{{\rm e}}
\def\d{{\rm d}}
\newcommand{\fer}[1]{(\ref{#1})}

\def\hat {\widehat }

\newcommand{\ve}{\varepsilon}
\newcommand{\va}{\varphi}

\newcommand{\RR}{\mathbb{R}}
\newcommand{\NN}{\mathbb{N}}

\newcommand{\bd}{\begin{displaymath}}
\newcommand{\ed}{\end{displaymath}}
\newcommand{\ba}{\begin{eqnarray}}
\newcommand{\ea}{\end{eqnarray}}

\title{On Rosenau-Type Approximations to \\  Fractional Diffusion Equations}

\author{G.Furioli \thanks{University of Bergamo, viale Marconi 5, 24044 Dalmine, ITALY.} \and A. Pulvirenti \thanks{Department of Mathematics,
University of Pavia, via Ferrata 1, 27100 Pavia, ITALY.} \and E. Terraneo \thanks{Department of Mathematics,
University of Milan, via Saldini 50, 20133 Milano, ITALY.} \and G. Toscani\thanks{Department of Mathematics,
University of Pavia, via Ferrata 1, 27100 Pavia, ITALY.} }

\begin{document}

\hyphenation{bounda-ry rea-so-na-ble be-ha-vior pro-per-ties cha-rac-te-ris-tic}

\maketitle

\begin{abstract}
Owing to the Rosenau argument \cite{Ros}, originally proposed to obtain a regularized
version of the Chapman-Enskog expansion of hydrodynamics, we introduce a non-local
linear kinetic equation which  approxi\-ma\-tes a fractional diffusion equation. We
then show that the solution to this approximation, apart of a rapidly vanishing in time perturbation, approaches the fundamental solution
of the fractional diffusion (a L\'evy stable law)  at large times.
     \end{abstract}

\vskip 3mm

{\small \noindent {\bf 2000 AMS subject classification.} 35K55, 35K60, 35K65, 35B40.
\


\

\section{Introduction}
In \cite{Ros}, Rosenau proposed a regularized version of the Chapman-Enskog expansion of hydrodynamics, with a suitably modified viscosity term.
 This model is given by the scalar equation
            \be\label{Rose}
              \partial_t f(v,t) +  \partial_v \Psi(f) (v,t)=  \mathcal D_\ve f(v,t),\quad v\in\R,\  t\geq 0
            \ee
                        where $\ve\ll 1$ is a small positive parameter,
            \[
            \widehat{\mathcal D_\ve f}(\xi,t) =   \frac {- \ve\xi^2}{1 + \ve^2 m^2 \xi^2}\hat f(\xi, t),
            \]
 and           $\hat g(\xi)$ denotes the Fourier transform of $g(v)$
        \[
       {\cal F}g(\xi)= \hat g(\xi) =\int_\R g(v) \e^{-iv\xi}\, \d v,\quad  \xi\in\R.
        \]

The operator on the right hand side looks like the usual viscosity
term $\ve f_{vv}$ at low wave-numbers $\xi$, while for higher wave
numbers it is intended to model a bounded approximation of a
linearized collision operator, thereby avoiding the artificial
instabilities that occur when the Chapman-Enskog expansion for such
an operator is truncated after a finite number of terms.

Note that the right hand side of \fer{Rose} can be written in the Fourier variable as
            \[
          \frac {- \ve \xi^2}{1 + \ve^2m^2 \xi^2}\hat f(\xi,t) = \frac \ve{(\ve m)^2 } \left (\frac 1 {1 +\ve^2 m^2 \xi^2}\hat f(\xi,t) -\hat f(\xi,t) \right) = {\cal F} \left(\frac 1{m \bar\ve} \left (N_{\bar\ve}*f -f\right)\right )(\xi, t)
            \]
 where $\bar\ve = m\ve$,  $*$ denotes convolution and
            \be
              \label{Max}
              N_\gamma(v) = \frac 1{2\gamma}\, \e^{-|v|/\gamma}
            \ee
is a non-negative function satisfying $\| N_\gamma\|_{L^1} = 1$. In
other words, the approximation proposed by Rosenau consists in
substituting  the linear diffusion equation
           \[
             \partial_t g(v,t) =  \partial_{vv}^2 g(v,t)
            \]
with the linear non-local kinetic equation
            \be
              \label{kin-heat}
              \partial_t g (v,t) = \frac {1}{\ve^2} \left[ N_{\ve}*g(v,t) -g(v,t)\right]
            \ee
in which the background ``Maxwellian'' $N_\ve$  is given by \fer{Max} \cite{Cer, PT13}.

The study of the main properties of the Rosenau approximation
\fer{Rose} and of its relaxation part \fer{kin-heat} has attracted a
lot of interest \cite{LT, ST, RT}. Indeed, while it is clear that at
fixed time the solution to \fer{kin-heat} represents, for
sufficiently small values of the $\ve$ parameter, a good
approximation of the solution to the heat equation, the
characteristics of its solution for large times, and its possible
similarities with that of the heat equation are not evident. This
last problem has been recently addressed and studied in details in
\cite{RT}.

In view of the recent results in \cite{RT}, the argument proposed by Rosenau for the linear heat equation appears to be of high interest for further applications. Maybe the most natural extension of his idea is to apply a similar modification to other types of linear diffusion equations, like the fractional diffusion equations.
Fractional in space diffusion equations share in fact with the linear diffusion a simple representation in Fourier variables, which is at the basis of the introduction of a suitable Rosenau approximation.  However, while interesting for its possible applications, this type of approximation  has not been studied so far.

Fractional in space diffusion equations appear in many contexts.
Among others, the review paper by Klafter et al. \cite{KZS97} provides numerous references to physical phenomena in which these anomalous diffusion occurs (cf. \cite{Cha98, BWM00, GM98, SBMW01, MFL02, Vaz11} and the reference therein for various details on both mathematical and physical aspects).

As it can be argued from the original application, the Rosenau approximation establishes a clear connection between diffusion equations and non-local kinetic equations.
Possible connections between Boltzmann type equations and fractional in space diffusion equations have been studied only recently in \cite{MMM11, FPTT12}. While the analysis of  Mellet, Mischler and Mouhot in \cite{MMM11} is devoted to the study of  linear
kinetic equations of Boltzmann type, and their connection with
fractional diffusion equations, the results in \cite{FPTT12} refer to the nonlinear one dimensional Kac model for dissipative collisions introduced in \cite{PT}, and to its grazing collision limit.

The fractional diffusion equations read
\begin{equation}\label{frac-eq}
\partial_tg(v,t)=-(\sqrt{-\Delta})^\lambda g(v,t),
\end{equation}
where $0 < \lambda <2$. The fractional derivative operator
$(\sqrt{-\Delta})^\lambda$ is defined in the Fourier variable as
\be\label{def-frac} {\cal F}\left((\sqrt{-\Delta})^\lambda  h\right
)(\xi)= |\xi|^\lambda \hat h (\xi).
 \ee
Similarly to \fer{Rose}, the Rosenau-type correction consists in
substituting the fractional diffusion equation (in Fourier variable)
 \be\label{frac-fou}
 \partial_t\hat g(\xi,t)=-|\xi|^\lambda \hat g(\xi,t)
 \ee
with the equation
 \be\label{Ros2}
\partial_t\hat g(\xi,t)=  \frac {- |\xi|^\lambda}{1 + |\ve\xi|^\lambda}\hat g(\xi,t) = \frac 1{\ve^\lambda} \left[ \widehat M_\lambda(\ve\xi)\hat g(\xi, t) - \hat g(\xi,t)\right],
 \ee
where $\ve\ll 1$.
Note that in this case the Maxwellian function is expressed in Fourier variable by
  \be\label{Max-f}
  \widehat M_\lambda(\xi) = \frac 1{1+ |\xi|^\lambda}.
  \ee
{It is notable that, for all $0<\lambda \leq 2$,  this Maxwellian
function is the characteristic function of a symmetric probability
distribution known in probability theory with the name of {\it
Linnik distribution} \cite{L, L2}.  In addition, when $\lambda >
1$, $\widehat M_\lambda \in L^1(\R)$, which allows us to apply the
inversion formula to conclude that $M_\lambda$ is a probability
density function. Consequently, in the case of a fractional
diffusion equation with $\lambda >1 $, Rosenau approximation
consists in substituting the fractional diffusion \fer{frac-eq}
with the non-local linear Boltzmann equation
 \be \label{kin}
 \partial_t g(v,t) = \frac {1}{\ve^\lambda} \left[ M_{\lambda,\ve}*g(v,t) -g(v,t)\right],
 \ee
where the Maxwellian $M_{\lambda,\ve}$ is defined through its
Fourier transform by the formula $\widehat
M_{\lambda,\varepsilon}(\xi)= \widehat M_{\lambda} (\eps \xi)$.
Unlikely we note that, on the contrary to what happens in the case
of the linear heat equation, where  the Maxwellian $N_\gamma$ is
explicitly given by \fer{Max}, in \fer{kin} the expression of the
Maxwellian \fer{Max-f} is no more explicit in the physical space.

This situation has evident analogies with the central limit
theorem of probability theory \cite{LR79}. Indeed, while the Rosenau
approximation of the heat equation is the analogue of the
classical central limit theorem, and the large-time behavior is
driven by a Gaussian density (the self-similar solution to the
heat equation), the approximation \fer{kin} is the analogue of
the central limit theorem for stable laws, where the expression of
the stable law is explicitly known only in the  Fourier variable.

Thanks to this analogy, the main features of the Maxwellian function
\fer{Max-f} can be extracted from classical results on the central
limit theorem for stable laws \cite{LR79, AA, KOA}. The distribution
function associated to the Maxwellian $M_\lambda$ belongs in fact to
the domain of normal attraction of the {\it L\'evy symmetric stable
distribution} of order $\lambda$, defined by
 \be\label{levy}
 \widehat L_\lambda(\xi) = \e^{-|\xi|^\lambda}.
 \ee
 We recall \cite{Fel,LR79} that  a distribution function $F$ belongs to the domain
of normal attraction of the stable law $L_\lambda(v) \d v$ if for
any sequence of independent and identically distributed
real-valued random variables $(X_n)_{n\geq 1}$ with common
distribution function $F$ there exists a sequence of real numbers
$(c_n)_{n\geq 1}$ such that the law of
\[
\frac{X_1+\dots+X_n }{n^{1/\lambda}}-c_n
\]
converges weakly to the stable law $L_\lambda(v) \d v$. We will
come back to this and other properties of the Linnik distributions
in Section \ref{general}  where we give a different
characterization of the domain of normal attraction in terms of
the asymptotic behavior at infinity. Without loss of generality, we will restrict our analysis to  centered distributions, so that $c_n=0$.

It is important to outline that the fractional diffusion
equation \fer{frac-eq} has a fundamental solution, given by the L\'evy distribution of order
$\lambda$. Indeed, the Fourier version \fer{frac-fou} is easily solved to give the solution
\be\label{fund}
\hat g(\xi,t) = \hat g_0(\xi)\e^{-|\xi|^\lambda t}.
\ee
In addition to the aforementioned difficulties related to the fact
that both Linnik distributions and the L\'evy distribution are not
explicitly  expressed in the physical space, a second main
difference with respect to the case of the heat equation is that in
this case the Maxwellian function has moments bounded only up to a
certain order. These facts mean that it is not immediate to assert
that the correction \fer{Ros2}, for a fixed time and a sufficiently
small $\ve$, is a \emph{good} approximation to the fractional
diffusion equation. Moreover, it is not clear whether or not the
large-time behavior of the solution to this approximation agrees
with the large-time behavior of the solution to the fractional diffusion.

The aim of this article is to give an answer to the previous
questions in the range $1< \lambda <2$ of fractional diffusion.
Like in the case of the heat equation, these results underline
that the Rosenau-type approximation can be viewed as a particular
case of a general approximation to the fractional diffusion
equation by means of a linear kinetic equation of type \fer{kin},
provided the \emph{background density} $M_\lambda$ is a probability
density function which lies in the domain of normal attraction of
the stable law \fer{levy}. In Theorem \ref{conveps}, it will be shown that,
in a certain metric equivalent to the weak*-convergence of
measures, the distance between the solution to the fractional
diffusion equation and the solution to the kinetic equation can be
bounded  in terms of $\ve$ and $t$.
While this bound allows us to obtain convergence for $\eps \to 0$, validating in this way
the approximating model,
the same bound is not enough to get convergence for large times, even in this
weak sense.
The solution of equation \eqref{kin}, which in the Fourier variable can be easily written as
\[
\hat g_\eps(\xi, t)=\e^{-t\eps^{-\lambda} (1-\widehat M_{\eps,\lambda}(\xi))} \hat g_0(\xi)
\]
can be expressed at least in a formal way, as a convolution in the physical space
\[
g_\eps(v,t)= P_{\lambda,\eps}(\cdot , t)\ast g_0(v).
\]
At difference with the solution to the fractional diffusion equation \fer{fund}, since $\widehat M_{\lambda, \eps}  \in C_0(\R)$  it follows that
$P_{\lambda, \eps}(v,t)$ cannot belong to $L^1(\R)$ at any positive time.
Indeed, as we will see in Section \ref{kinetic} resorting to a suitable expansion
\[
 P_{\lambda, \eps}(v, t) ={\rm e}^{- \ve^{-\lambda} t} \delta_0+{\rm
e}^{- \ve^{-\lambda} t}\sum_{n=1}^{\infty}\left(\frac t{\eps ^\lambda}\right)^n
\frac1{n!}M_{\lambda,\eps}^{*n}(v).
\]
one identifies in the above expression  a singular part $\displaystyle{{\rm e}^{- \ve^{-\lambda} t} \delta_0}$,
which vanishes exponentially fast both for $\eps \to 0$ and  $t\to +\infty$, and a regular  part belonging to $L^1(\R)$.
For the particular case when $M_{\lambda}$ is a Linnik distribution \eqref{Max-f}, we  will be able to
recover the long time asymptotic behavior in $L^1(\R)$ at the price of discarding this singular part. This will be done by introducing a
regularized approximated solution, in which the convolution kernel  $P_{\lambda, \eps}$ is replaced by
\[
P_{\lambda, \eps,reg}(v,t)=P_{\lambda,\eps}(v,t)+{\rm
e}^{-\ve^{-\lambda} t}\left(
M_{\lambda,\eps}(v) - \delta_0(v)\right),
\]
obtained by substituting the ${\rm e}^{- \ve^{-\lambda} t} \delta_0 \ast g_0$ term by ${\rm e}^{- \ve^{-\lambda} t}M_{\lambda, \eps} \ast g_0$.

The plan on the article is then as follows. In Section \ref{frac-diff} we list various properties of the fractional diffusion equation. In particular, we prove convergence in $L^1(\R)$
for large times to the fundamental solution \eqref{fund}.
 In Section \ref{kinetic},
using tools of the kinetic theory of rarefied gases, we
introduce a possible derivation and the main features of the
Rosenau approximation with a general Maxwellian belonging to the
domain of normal attraction of the stable law \fer{levy}. This
kinetic formulation is used to obtain explicit solutions to
the Rosenau equation \eqref{kin}, using both Fourier transform and
Wild sums \cite{PT13}. In Section \ref{general} we
deal with the problem of approximating the fractional
diffusion at any fixed finite time with the Rosenau-type equation
\fer{kin} with a general Maxwellian.
Last, in Section \ref{main-sect}, we investigate the large
time behavior of the solutions to \eqref{kin} with a Linnik
distribution as a Maxwellian. We first show that
with a suitable time-scaling, the approximated solution
 behaves asymptotically in time as the solution of the fractional diffusion equation
  in a suitable Fourier-based metric ({see next subsection}).
As a consequence, both scaled solutions converge in the same metric towards the asymptotic profile of the fundamental solution.
Then, we show that after a suitable regularization of the Rosenau equation, obtained by
discarding its singular part,  the approximated solution  behaves  asymptotically  in $L^1(\R)$ as the solution of the fractional diffusion equation.

{The case $\lambda=1$ is a special case, since $\widehat M_1(\xi)=
(1+|\xi|)^{-1}$ is not in $L^1(\R)$ and we did not succeed in performing
analogous calculations as in the cases $1<\lambda <2$. However, the same results can be obtained with minor
modifications provided the Linnik distribution  $\widehat M_1(\xi)$
is replaced in equation \eqref{kin} by the L\' evy distribution
$\widehat L_1(\xi)= \e^{-|\xi|}$ itself.}

These technical results possess an evident interest for a consistent  numerical approximation of the fractional diffusion equation \fer{frac-eq}. Indeed, let us start with a probability density function $g(v, t=0)= g_0(v)$. Then a semi-implicit Euler scheme applied to the kinetic equation \fer{kin} shows that, in a fixed small time-interval $\Delta t$, the solution can be approximated according to the rule
 \be\label{euler}
 g(v,t +\Delta t) = \frac{\ve}{\ve + \Delta t}g(v,t) + \frac{\Delta t}{\ve + \Delta t}M_{\lambda,\ve}*g(v,t).
 \ee
Therefore, at the time step $n+1$, the solution is obtained by a convex combination of the solution at the time step $n$, and of the convolution between the solution at the time step $n$ with the constant Maxwellian $M_{\lambda, \eps}$. In particular, expression \fer{euler} can be easily implemented by Monte Carlo methods \cite{PT13}.

\subsection{{Functional framework}}
      \label{subFuncFram}

Before entering  into the main topic of this paper, we list below the
various functional spaces, distances and norms used in our analysis. For $p \in [1,
+\infty)$ and $q \in [1, +\infty)$, we denote by $L_q^p$  the
weighted Lebesgue spaces
 \begin{equation*}
L^p_q(\RR) := \left\{ f : \mathbb{R} \rightarrow \mathbb{R} \text{
measurable; }\|f\|_{L^p_q}^p := \int_{\mathbb{R}} |f(v)|^p \, (1+
v^2)^{q/2} \, \d v < +\infty \right\}.
 \end{equation*}
 In particular, the usual Lebesgue spaces are given by \[ L^p := L^p_0.\]
Moreover, for $f \in L^1_q$, we can define for any $\alpha\leq q$
the $\alpha^{th}$ order \emph{moment} of $f$ as the quantity
 \begin{equation*}
   m_\alpha(f):= \int_\RR f(v) \, |v|^{\alpha} \d v \, < + \infty.
 \end{equation*}
For $s \in \NN$, we denote by $W^{s,p}$ the Sobolev spaces
 \begin{equation*}
W^{s,p}(\RR) := \left\{ f \in L^s; \|f\|_{W{s,p}}^p := \sum_{|k|
\leq s} \int_{\RR} \left |f^{(k)}(v)\right |^p \, \d v < + \infty
\right \}.
 \end{equation*}
If $p=2$ we set $H^s := W^{s,2}$.

Given a probability density $f$, we define its {Fourier
transform} $\mathcal F_v(f)$ by
 \begin{equation*}
\mathcal F_v(f)(\xi) = \widehat{f}(\xi) := \int_\RR f(v) e^{- i
\xi  v} \, \d v, \qquad \xi \in \RR
 \end{equation*}
 and the inverse Fourier transform as
 \[
 \varphi^{\vee}(v)= \frac 1{2\pi} \int_\R \varphi(\xi) \e^{i\xi v} \, \d \xi.
 \]
The Sobolev space $H^s$ can equivalently be defined for any $s
\geq 0$ by the norm
 \begin{equation*}
\|f\|_{H^{s}} := \left \| \mathcal F_v \left (f \, \right )\right
\|_{L^2_{2s}}.
 \end{equation*}
The homogeneous Sobolev space $\dot H^s$ is then defined by the
homogeneous norm
 \begin{equation*}
\|f\|_{\dot H^s}^2 := \int_\RR |\xi|^{2 s}| \widehat{f}(\xi)
 |^2 \, \d \xi.
 \end{equation*}

Finally, we introduce a family of Fourier-based metrics in the
following way: given $s >0$ and two probability densities $f$ and
$g$, their Fourier-based distance $d_s(f,g)$ is the quantity
 \begin{equation*}
d_s(f,g) := \sup_{\xi \neq  0} \frac{\left
|\widehat{f}(\xi)-\widehat{g}(\xi)\right |}{|\xi|^s}.
 \end{equation*}
This distance is finite, provided that $f$ and $g$ have finite moment of order  $s$ and
\[
\int v^k f(v) \d v= \int v^k g(v)\d v, \quad k=1,2, \dots, [s]
\]
where, if $s \notin \NN$,  $[s]$ denotes
the entire part of $s$ (or up to order $s-1$ if $s \in \NN$).
Moreover $d_s$ is an \emph{ideal} metric \cite{CT07}. Its main
properties are the following
\smallskip
 \begin{itemize}
\item[i)] For all probability densities $f$, $g$, $h$,
\[ d_s(f* h, g*h) \leq d_s (f, g); \]
\item[ii)]
Define for a given nonnegative constant $a$ the dilation \[f_a(v) =
 \frac 1a f\left ( \, \frac va \, \right ).\]
Then for any pair of probability densities $f$, $g$, and any
nonnegative constant $a$
 \be \label{scala}
  d_s( f_{a}, g_{ a}) = a^s \, d_s(f,g).
  \ee
\end{itemize}
The $d_s$-metric is related to other more known metrics of large use
in probability theory \cite{GTW95}. In parti\-cular, two classical
interpolation inequalities (see \cite{CGT} for proofs) will be used
in the following:
 \be
\begin{aligned}\label{ine!}
\left\| f\right\|_{L^1}&\leq C\left\|
f\right\|_{L^2}^{\frac{2\alpha}{1+2\alpha}}
m_{\alpha}(f)^{\frac{1}{1+2\alpha}},\qquad \alpha>0\\
\left\| f\right\|_{L^2}&\leq C\left( \sup_{\xi\neq 0}\frac{|\hat
f(\xi)|}{|\xi|^s}\right )^{\frac{2s}{1+4s}} \left\|f\right\|_{\dot
H^s}^{\frac{1+2s}{1+4s}},\qquad s>0.
\end{aligned}
\ee

\section{The fractional diffusion equation}\label{frac-diff}
We recall here the existence result for the Cauchy
problem associated to the fractional diffusion equation of order
$\lambda$, $0<\lambda <2$, with initial data $g_0\in L^1(\R)$,
\begin{equation}\label{fraceq}
\left \{
\begin{aligned}
& \partial_tg(v,t)=-(\sqrt{-\Delta})^\lambda g(v,t), \ \ \ \ \ \
v\in \RR, \ t> 0\\
&g(v,0)= g_0(v)
\end{aligned}
\right .
\end{equation}
and we briefly list properties of the solution that are used in our analysis.

By considering equation \eqref{fraceq} in  the Fourier variable it is
straightforward to show that for any initial data $g_0 \in L^1$
there exists a  unique solution
 \be\label{sol1}
g(v,t)=\frac{1}{t^{1/\lambda}}L_\lambda\left(\frac{\cdot}{t^{1/\lambda}}\right)\ast
g_0(v),
  \ee
 where
$L_\lambda$
 is the L\'evy distribution of order $\lambda$ defined by
 $\hat L_{\lambda}(\xi)={\rm e^{-|\xi|^\lambda}}$. For the sake of
 simplicity we denote the fundamental solution of the fractional
 diffusion equation by
 \[
 P_\lambda(v,t)=\frac{1}{t^{1/\lambda}}L_\lambda\left(\frac v{t^{1/\lambda}}\right),
 \]
 and so
 \[
 g(v,t)=P_\lambda(\cdot,t)\ast g_0(v).
\]

 To outline the analogies
between the present problem and the classical central limit
theorem for stable laws \cite{LR79}, we will further assume in the
rest of the paper that $g_0$ is a probability density function. By
 mass conservation, the solution $g(t)$ will remain a
probability density for all subsequent times.

 It is well known \cite{LR79} that in the interval $0<\lambda<2$, $L_\lambda$ belongs to  $C_0(\R)$ and it is an even probability density.
 For $\lambda =1$ the L\'evy symmetric stable distribution coincides with the Cauchy distribution
 \[
 L_1(v)= \frac 1 \pi \frac 1{1+v^2}.
\]
 Even though for $\lambda \neq 1$ $L_\lambda(v)$ is not known explicitly in the physical variable, its behavior for large $v$  is known in details  \cite{Pol}. It holds
\[
L_\lambda(v) \sim {\frac 1 {\pi} \Gamma(1+\lambda) \sin \left(
\frac {\pi\lambda}2\right )}{|v|^{-(1+\lambda)}}, \quad |v|\to
+\infty.
\]
In consequence the heavy tailed density $L_\lambda$  has bounded moments only up to some order
 \be\label{moments}
\int_{\R} |v|^\alpha L_\lambda (v) \,  \d v =
m_\alpha(L_\lambda)<+\infty,\quad 0<\alpha <\lambda.
\ee
The case
$\lambda=2$ in equation \eqref{fraceq} corresponds to  the heat
equation
\[
\left\{
\begin{aligned}
&\partial_tg(v,t)=\Delta g(v,t)\ \ \ \ v\in \R, \ t>0\\
&g(v, 0)= g_0(v)\in L^1(\R).
\end{aligned}
\right.
\]
In this case it is well known that  the explicit solution is
\[
g(v,t)= \Omega(\cdot,t) \ast g_0(v)
\]
where $\displaystyle{\Omega(v,t) = \frac{1}{\sqrt{4\pi t}} \e^{-\frac{v^2}{4t}}}$
is the {\it heat kernel}.
It can be shown \cite{BB,To96} that $g(t)$ behaves as the heat kernel as $t\to +\infty$, provided that $g_0$ is a probability density
of finite energy and entropy, namely
\[
\int_\R v^2\, g_0(v) \d v<+\infty,\quad \int_\R | \ln g_0(v)| g_0(v)\, \d v<+\infty
\]
and more precisely the following bound is known to be sharp
\[
\| g(t)- \Omega(t)\|_{L^1} \leq \frac C{\sqrt{1+2t}}, \quad t > 0
\]
where $C$ is an explicit constant.

A similar behavior occurs for  the solution of the fractional
diffusion equation, but in this case it appears difficult to obtain an explicit rate of
approximation. However one can state the following
proposition.
\begin{prop}
\label{prop1} Let $g(t)$ be the solution of the Cauchy problem \eqref{fraceq},
corresponding to the initial value $g_0$,  a probability density function. Then
\[
\lim_{t\to +\infty}\left\| g(t) -
\frac{1}{t^{1/\lambda}}L_\lambda\left(\frac{\cdot}{t^{1/\lambda}}\right)\right\|_{L^1}
=0.
\]
\end{prop}

\begin{proof} As in the case of the heat equation,  this convergence
can be  obtained by passing from the fractional diffusion equation
\eqref{fraceq} to the corresponding Fokker--Planck equation
\be\label{fp}
\partial_t u(v,t)= -( \sqrt {-\Delta} )^\lambda u(v,t) +\frac 2 \lambda \partial_v ( v u(v,t) )
\ee
through the change of unknown function
\be
u(v,t)= \e^{\frac 2\lambda t} g\left(\e^{\frac 2 \lambda t} v, \frac {\e^{2t}-1}2\right)
\ee
or, in an equivalent way,
\be
g(w, \tau)= \frac1{\left(\sqrt{2\tau+1}\right )^{2 /\lambda}} u\left( \frac w{\left(\sqrt{2\tau+1}\right )^{2 /\lambda}} , \frac {\ln(2\tau +1)}2\right ).
\ee
If we write equations \eqref{fraceq} and \eqref{fp} in the Fourier variable
\[
\begin{aligned}
&\partial_t \hat g(\xi,t)= -|\xi|^\lambda \hat g(\xi, t)\\
&\partial_t \hat u(\xi, t)= -|\xi|^\lambda \hat u(\xi, t) -\frac 2 \lambda\xi \partial_\xi \hat u(\xi, t)
\end{aligned}
\]
and the explicit expression of the Fourier transform of the solution of the fractional diffusion equation with $g_0$ as initial data  as ${\dis \hat g(\xi,t)= \e^{-|\xi|^\lambda t} \hat g_0(\xi)}$, we get
\[
\hat u(\xi, t)= \exp \left({-|\xi|^\lambda \frac {1-\e^{-2t}}2} \right )\hat g_0\left( \e^{-{\frac 2 \lambda} t }\xi\right ).
\]
In \cite{FPTT12}, Proposition 3, it was proved that for a given probability density
$g_0$ we have
\[
\lim_{t\to +\infty} \left\| u(t) - {\cal F}^{-1} \left( \exp \left(-\frac {|\xi|^\lambda}2\right )\right ) \right \|_{L^1}=0.
\]
This implies
\[
\lim_{\tau\to +\infty} \left\| u\left(\cdot, \frac {\ln(2\tau+1)}2\right ) - {\cal F}^{-1} \left( \exp \left(-\frac {|\xi|^\lambda}2\right )\right ) \right \|_{L^1}=0.
\]
By the scaling invariance of the $L^1$ norm we get
\[
\lim_{\tau\to +\infty} \left\| \frac 1{\left(\sqrt{2\tau+1}\right )^{{2 /\lambda}}}
u\left(\frac{\cdot}{\left(\sqrt{2\tau+1}\right )^{{2 /\lambda}}} , \frac
{\ln(2\tau+1)}2\right ) - {\cal F}^{-1} \left( \exp\left(-\frac
{\left|\left(\sqrt{2\tau+1}\right )^{{2 /\lambda}} \xi\right |^\lambda}2\right
)\right ) \right \|_{L^1}=0.
\]
Therefore
\[
\lim_{t\to +\infty} \left\| g(t) - {\cal F}^{-1} \left( \exp\left(-\frac {(2t+1)}2
|\xi|^\lambda\right )\right ) \right \|_{L^1}=0.
\]
In order to get the desired result, it is enough to verify that
\[
\lim_{t\to +\infty}\left\| g(t) -
\frac{1}{t^{1/\lambda}}L_\lambda\left(\frac{v}{t^{1/\lambda}}\right)\right\|_{L^1} =
\lim_{t\to +\infty} \left\| g(t) - {\cal F}^{-1} \left( \e^{-t {|\xi|^\lambda}}\right
) \right \|_{L^1} =0.
\]
This follows since
\[
\lim_{t\to +\infty} \left\| {\cal F}^{-1} \left( \e^{-t {|\xi|^\lambda}}\right )   -  {\cal F}^{-1} \left( \exp\left(-\frac {(2t+1)}2
|\xi|^\lambda\right)\right )  \right \|_{L^1}=0.
\]
Indeed
\[
\begin{aligned}
&{\cal F}^{-1} \left( \e^{-t {|\xi|^\lambda}}\right ) = \frac
1{t^{ 1/ \lambda}} L_\lambda \left( \frac v{t^{ 1/{\lambda}}}
\right)\\
& {\cal F}^{-1} \left( \exp\left(-\frac {(2t+1)}2
|\xi|^\lambda\right)\right ) = \frac 1{\left(t+{1 /2}\right)^{
1/ \lambda}} L_\lambda \left(\frac v{\left(t+{1 /2}\right )^{
1/{\lambda}}} \right)
\end{aligned}
\]
with $L_\lambda$ the L\'evy distribution, with  $\displaystyle{L_\lambda (v) \leq
\frac {C_\lambda}{1+|v|^{\lambda +1}}}$ for all  $v\in\R$.
\end{proof}

Another important argument concerned with the solution of the fractional diffusion
equation is the evolution of moments which are initially bounded. We already noticed
in \eqref{moments} that the moments $m_\alpha (L_{ \lambda})$ of $L_{\lambda}$ are
bounded for $0<\alpha <\lambda$. Let us consider the Cauchy problem \eqref{fraceq}
with initial data $g_0$,  a probability density with bounded  moments in the range
$0<\alpha<\lambda$,
\[
\int_\R |v|^\alpha g_0(v)\, \d v =
m_{\alpha} (g_0)<+\infty.
 \]
 Since
\[
g(v,t) = P_\lambda(\cdot,t)\ast g_0(v)=\frac 1{t^{1
/\lambda}}L_\lambda \left(\frac \cdot {t^{ 1/ \lambda}}\right )
\ast g_0(v)
\]
and
\be\label{mom-plambda}
\int_\R |v|^\alpha   P_\lambda(v, t) \d v=   \int_\R |v|^\alpha  \frac 1{t^{1/ \lambda}}L_\lambda \left(\frac v {t^{1 /\lambda}}\right )
 \d v = t^{{\alpha /\lambda}} m_\alpha\left( L_{\lambda}\right)
\ee
one obtains
\[
\begin{aligned}
&\int_\R |v|^\alpha P_\lambda(\cdot, t)\ast g_0 (v) \d v = \int_\R\int_\R |v|^\alpha  P_\lambda (v-w,t)g_0(w) \d w \d v\\
&= \int_\R\int_\R |v-w+w|^\alpha  P_\lambda (v-w,t)g_0(w) \d w \d v.
\end{aligned}
\]
Let $0<\alpha \leq 1$. Then we get
 \be\label{mom-alphapiccolo}
\begin{aligned}
& \int_\R\int_\R |v-w+w|^\alpha  P_\lambda (v-w,t)g_0(w) \d w \d v \leq \int_\R\int_\R \left(|v-w|^\alpha+|w|^\alpha\right )  P_\lambda (v-w,t)g_0(w)
 \d w \d v\\
&= t^{{\alpha /\lambda}} m_\alpha\left( L_{ \lambda}\right)+
m_{\alpha}(g_0).
\end{aligned}
\ee
 If $1<\alpha <\lambda$
\be
\begin{aligned}\label{mom-alphagrande}
& \int_\R\int_\R |v-w+w|^\alpha  P_\lambda (v-w,t)g_0(w) \d w \d v \leq 2^{\alpha -1} \int_\R\int_\R \left(|v-w|^\alpha+|w|^\alpha\right )  P_\lambda (v-w,t)g_0(w) \d w \d v\\
&= 2^{\alpha-1} \left(t^{{\alpha /\lambda}} m_\alpha\left( L_{
\lambda}\right)+ m_{\alpha}(g_0)\right ).
\end{aligned}
 \ee
 In both cases, the moments of the solution are uniformly bounded above by
 an explicit function of  time which grows as $t^{\lambda/\alpha}$.

We end this Section by proving that, in complete analogy with the heat equation,  any
convex functional is non-increasing along the solution to the fractional diffusion
equation. First of all, we remark that for $t_2>t_1>0$
 \[
\widehat P_\lambda (\xi,t_2) = \widehat L_\lambda(\xi t_2^{ 1/\lambda})=
\e^{-|\xi|^\lambda t_2} = \e^{-|\xi|^\lambda (t_2-t_1)} \e^{-|\xi|^\lambda t_1}  =
\widehat P_\lambda (\xi,t_2-t_1) \widehat P_\lambda (\xi,t_1).
\]
Owing to  expression   \fer{sol1} for the solution we obtain, for $t_2>t_1>0$,
\[
g(v,t_2)=P_\lambda (\cdot ,t_2)\ast g_0(v)=P_\lambda (\cdot ,t_2-t_1)\ast g(\cdot ,t_1)(v).
\]
Now, let $\phi(r)$, $ r \ge 0$ be a (smooth) convex function of $r$ and consider the
functional
\[
 \Phi(g)(t)= \int_\R \phi (g(v,t))\, \d v.
 \]
If $t_2>t_1>0$, we get
\[
\begin{aligned}
\Phi(g)(t_2)=
 \int_{\R} \phi(g(v, t_2)) \, \d v = \int_\R \phi \left(\int_\R P_\lambda (w,t_2-t_1)g(v-w, t_1) \d w\right ) \d v.
\end{aligned}
 \]
 Now, use the fact that $P_{\lambda}$ is a probability density, so
that  by Jensen's inequality
 \[
 \begin{aligned}
\int_{\R}
\phi\left(\int_{\R}g(v-w, t_1)P_{\lambda}(w,t_2-t_1)\, \d w \right)\,\d v  &\le
\iint_{\R^2} \phi(g(v-w, t_1)) P_{\lambda}(w, t_2-t_1)\, \d v \,\d w\\
& = \int_{\R} \phi(g(v,t_1))\, \d v= \Phi(g)(t_1).
\end{aligned}
 \]
Hence $\Phi(g)$ is non-increasing.


\section{The linear kinetic equation}\label{kinetic}

In this Section, we will briefly discuss both the derivation and the main properties
of the linear kinetic equation \fer{kin}. Let us consider a system composed of many
identical particles. Let the number of particles with velocity $v$ at time $t$ be
described by the process $X(t)$  with probability density $g(t)$, and suppose that the
variation of $X(t)$ is solely due to interaction with an external background. The
background $B_\lambda$ is here described by a random variable with probability density
$M_\lambda$, which we will assume in the domain of attraction of a L\'evy distribution
of order $\lambda$, given by \fer{levy}. Let us further assume that the interaction
process of a particle with velocity $v$ with a background particle with velocity $w$
generates a post-interaction velocity $v^*$ given by
\begin{equation} \label{collProc}
v^* = v + w .
\end{equation}
In terms of the process $X(t)$ the law of change given by
\fer{collProc} can be rewritten as \index{Rosenau approximation random process}
 \[
 X^*(t) = X(t) + B_\lambda,
 \]
which implies, in the case in which $B_\lambda$ and $X(t)$ are independent,
that the law of $X^*(t)$ is the convolution of the laws of $B_\lambda$ and
$X(t)$. Assuming that $X(t), B_\lambda$ are independent each other, for a
given observable quantity $\va(\cdot)$, we then have that the mean
value of $\va(X)$ satisfies

 \be\label{kine-w}
\frac \d{\d t}\int_{\R} \phi(v)g(v,t) \, \d v = \frac \d{\d
t}\langle \va(X(t))\rangle = \sigma \iint_{\R \times \R}\left(
\va(v^*) - \va(v)\right)\,g(v,t) M_\lambda(w) \, \d v\, \d w ,
 \ee
where the constant $\sigma> 0$ denotes as usual the interaction frequency. Note that
choosing $\varphi(v) = 1$ one shows that, independently of the background
distribution, $g(t)$ remains a probability density if it so initially
\[
\int_{\RR} g(v,  t)\,\d v = \int_{\RR} g_0(v)\,\d v = 1 .
\]
This is in general the unique conservation law associated to
equation \fer{kine-w}.

The effects of the background can be easily  modulated by considering, for a given
small positive parameter $\ve$, the random variable $\ve B_\lambda$.  To emphasize
this dependence, we will denote its distribution as
\[
M_{\lambda,\ve}(w)= \ve^{-1} M_\lambda (\ve^{-1} w).
\]
Then, by inserting $M_{\lambda,\ve}(w)$ into \fer{kine-w}, and setting the interaction
frequency $\sigma= 1/\ve^\lambda$, the kinetic equation \fer{kine-w} coincides with
the Rosenau approximation \fer{kin} in weak form. Hence, the Rosenau approximation of
a fractional diffusion equation of order $\lambda$ describes a system of particles
which modify their distribution through interactions with a background distributed
according to a probability law in the domain of attraction of a L\'evy stable law of
order $\lambda$.

\subsection{Representations of the solution to the Rosenau approximation}

The  Rosenau approximated equation
\begin{equation}\label{approx}
 \partial_t g_\varepsilon(v,t)
 =\frac{1}{\varepsilon^\lambda}\left[M_{\lambda, \varepsilon} \ast g_\eps (v, t)  -g_\eps(v,t)\right ], \quad \eps \ll 1
\end{equation}
 where $\widehat M_{\lambda,\varepsilon}(\xi)= \widehat M_{\lambda} (\eps \xi)$ and
 $ M_{\lambda}$ belongs to the domain of normal attraction of the
 stable law (\ref{levy}),
 is a linear non local  kinetic equation of Boltzmann type. Existence results for this
 equation are well-established.
To find a solution of the Cauchy problem with $g_0\in L^1$ as initial data, we can
resort to  two equivalent methods. Resorting to the equation in the Fourier variable
one can get  a first explicit representation of the solution. This solution can be
expressed as
 \be\label{p_epslambda-spazio}
 g_\eps(v,t) = P_{\lambda, \eps}(\cdot , t) \ast g_0(v),
 \ee
 where, in the Fourier variable
 \be\label{p_epslambda-fur}
 \widehat P_{\lambda, \eps}(\xi, t)=
{\rm e}^ {-\varepsilon^{-\lambda} t\left(1- \widehat M_{\lambda,\eps} (\xi)\right)}.
 \ee
We underline that, since $M_\lambda\in L^1(\R)$ } for every $t
> 0$,  $\widehat P_{\lambda, \eps}(\xi, t) \notin C_0(\R)$ and consequently
$P_{\lambda, \eps}(v, t) \notin L^1(\R)$. This (unpleasant) feature of $P_{\lambda,
\eps}(\cdot, t)$ will appear in a clearer way by applying the so-called Wild sum
expansion.

This expansion allows a useful  representation of the solution of equation
\eqref{approx}. It has been first introduced by Wild to construct a solution to the
Boltzmann equation for Maxwell molecules \cite {Wil51}, and it appears well adapted to
both linear and nonlinear kinetic equations \cite{PT13}. Let $h_\varepsilon(v,t)={\rm
e}^{ \varepsilon^{-\lambda}t }g_\varepsilon(v,t)$. Then the Cauchy problem associated
to equation \eqref{approx} can be written as a fixed point problem as follows. Since

\[ \left\{
\begin{aligned}
&\left(\partial_tg_\varepsilon(v,t)+\frac{1}{\varepsilon^\lambda}g_\varepsilon(v,t)\right){\rm
e}^{\ve^{-\lambda} t}=\frac{{\rm e}^{\ve^{-\lambda}
t}}{\varepsilon^\lambda}M_{\lambda,\eps}\ast
g_\eps(\cdot,t)(v)\\
&g_\varepsilon(v,0)=g_0(v),
\end{aligned}\right.
\]
 then
$$
g_\eps(v,t){\rm e}^{\ve^{-\lambda} t}-g_\eps(v,0)=\frac 1{\eps^\lambda}\int_0^t{\rm
e}^{\ve^{-\lambda}s} M_{\lambda,\eps}\ast g_\eps(v,s)\, \d s.
$$
Therefore
$$
h_\eps(v,t)=g_0(v)+\frac1{\eps^\lambda}\int_0^t
M_{\lambda,\eps}\ast h_\eps(v,s)\, \d s=\Phi_\eps(h_\eps)(v,t).
$$
Starting from
$$
h_\eps^{(0)}(v)=g_0(v)$$
 and defining for any $n\geq 0$
 $$
h_{\eps}^{(n+1)}(v,t)=\Phi_\eps(h_\eps^{(n)})(v,t)
$$
 we construct the monotone sequence
$$
h_\eps^{(n)}(v,t)=h_\eps^{(n-1)}(v,t)+\left({t\eps^{-\lambda}}\right)^n\frac{1}{n!}
\ M^{*n}_{\lambda,\eps}\ast g_0(v),
$$
which converges in $L^1(\R)$ towards
$$
h_\eps(v,t)=g_0(v)+\sum_{n=1}^{\infty}\left(\frac t{\eps ^\lambda}\right)^n
\frac1{n!}M_{\lambda,\eps}^{*n}\ast g_0(v).
$$
Finally, we obtain the expression
 \be\label{wi2}
g_\eps(v,t)={\rm e}^{- \ve^{-\lambda} t} g_0(v)+{\rm e}^{-
\ve^{-\lambda} t}\sum_{n=1}^{\infty}\left(\frac t{\eps
^\lambda}\right)^n \frac1{n!}M_{\lambda,\eps}^{*n}\ast g_0(v).
 \ee
By comparing \fer{wi2} with \eqref{p_epslambda-spazio} one obtains an explicit
representation of the fundamental solution $P_{\lambda, \eps}(\cdot, t)$
 \be\label{fu2}
 P_{\lambda, \eps}(v, t) ={\rm e}^{- \ve^{-\lambda} t} \delta_0(v)+{\rm
e}^{- \ve^{-\lambda} t}\sum_{n=1}^{\infty}\left(\frac t{\eps ^\lambda}\right)^n
\frac1{n!}M_{\lambda,\eps}^{*n}(v).
 \ee
At difference with the fundamental solution of the original fractional diffusion
equation \fer{frac-eq}, expression \fer{fu2} shows that $P_{\lambda, \eps}(\cdot, t)$
contains a singular part, the Dirac delta function $\delta_0$ located in $v=0$, of
size exponentially decaying with time and $\eps$.

\subsection{Properties of the solution to the Rosenau approximation}
 Equation \fer{kine-w} allows us to control
the time evolution  of the moments of $g(t)$. For a given constant $\alpha >0$, let us
take $\varphi(v) = |v|^\alpha$. We obtain
\[
   \frac{\d}{\d t} \int_{\R} |v|^\alpha g(v, t) \, \d v =  \,
\frac 1{\ve^\lambda} \iint_{\RR^2} \left[ |v^*|^\alpha - |v|^\alpha
\right] g(v, t)M_{\lambda,\ve}(w) \, \d v \, \d w.
\]
Since $|v^*|^\alpha = | v+ w|^\alpha \le c_\alpha (|v|^\alpha
+|w|^\alpha)$, the moment of $g(t)$ of order $\alpha$ is bounded if the corresponding moment of the background distribution is
bounded.
Proceeding as in  Lemma \ref{lemma4} one  shows
 that the evolution of the moments  is polynomial in time, in perfect agreement with the evolution of the
corresponding moments for the solution of the fractional diffusion equation, as given
by  \eqref{mom-alphapiccolo}, \eqref{mom-alphagrande} (see Remark \ref{rem7}).

Having in mind the discussion in \cite{RT}, a further interesting analogy with the
linear diffusion is found by looking at the evolution of convex functionals along the
solution. Let $\phi(r)$, $ r \ge 0$ be a (smooth) convex function of $r$ and consider
\[
 \Phi(g)(t)= \int_\R \phi (g(v,t))\, \d v.
 \]
 Then, using equation  \fer{kin} we obtain
 \[
 \frac{\d}{\d t}\Phi(g) (t) =
\frac{\d}{\d t} \int_{\R} \phi(g(v, t)) \, \d v =  \int_{\R} \phi'
(g(v, t)) \frac{\partial g(v, t)}{\partial t}  \, \d v =
 \]
 \[
\frac{1}{\ve^\lambda} \int_{\RR} \phi' (g(v, t)) \left(
M_{\lambda,\ve}*g(v,t)- g(v,t)\right) \, \d v .
 \]
Thanks to the convexity of $\phi$,  for $r,s \ge 0$
 \[
\phi'(s) (r-s) \le \phi(r) - \phi(s),
 \]
and one obtains
\[
\frac{\d}{\d t} \int_{\R} \phi(g(v, t)) \, \d v \le
\frac{1}{\ve^\lambda} \int_{\R} \left( \phi( M_{\lambda,\ve}*g(v,t)) -
\phi( g(v,t)) \right) \, \d v.
 \]
Now, use the fact that $M_{\lambda,\ve}$ is a probability density, so
that  by Jensen's inequality
 \[
\int_{\R}\phi( M_{\lambda,\ve}*g(v,t)) \, \d v =  \int_{\R}
\phi\left(\int_{\R}g(v-w,t) M_{\lambda,\ve}(w)\, \d w \right)\,\d v  \le
 \]
 \[
\int_{\R^2} \phi(g(v-w,t)) M_{\lambda,\ve}(w)\, \d v \,\d w = \int_{\R}
\phi(g(v,t))\, \d v
 \]
 to conclude
\[
\frac{  \d}{\d t} \int_{\RR} \phi(g(v, t)) \, \d v \le 0.
 \]
Thus any convex functional is non-increasing along the solution to
the Rosenau type kinetic equation \fer{kin}.

\section{An approximation result}\label{general}

As briefly discussed in the Introduction, one of the main novelties that can be
extracted by the Rosenau approximation is that the kinetic model \fer{approx} has an
evident interest from the point of view of its numerical approximation. This feature
has been extensively investigated in the case of the linear diffusion in \cite{RT},
where it has been shown that the linear kinetic model represents a consistent
approximation of the heat equation even if the Maxwellian density generated by the
Rosenau idea is substituted by a different one, provided that some properties about
moments are fulfilled. One of the results of this investigation has been the inclusion
of a singular Maxwellian producing the central difference scheme, among the admissible
Maxwellian densities for the corresponding linear kinetic model. Trying to get a
similar result for the problem under consideration, we consider  a linear kinetic
equation of type \eqref{approx}
\[
\partial_t g_\varepsilon(v,t) =\frac 1{\varepsilon^\lambda}\left[M_{\lambda,\varepsilon}\ast g_\varepsilon(v,t)
-g_\varepsilon(v,t)\right],
\]
in which the Linnik density \eqref{Max-f} is replaced by a Maxwellian $M_\lambda$ with the
properties to be even and  the density of a centered distribution function belonging
to the domain of normal attraction of the stable law $L_\lambda(v)\d v$, of exponent
$\lambda\in (1,2)$. The evident advantage to work with a Maxwellian density different
from Linnik distribution is that, as we will see, one can resort to a density which is
explicitly known in the physical space.

In the rest of this section, we aim at proving that, with a suitable choice of
distance, at any fixed time the solutions to \fer{approx} and to the fractional
diffusion equation, are getting closer in terms of the parameter $\ve$. First of all
we underline that it is not even clear in which sense the solution to \eqref{approx}
represents  an approximation to the solution of the fractional diffusion equation
\eqref{frac-eq}  as $\eps$ tends to zero. This is in contrast with what happens to the
original Rosenau approximation to the heat equation \eqref{kin-heat}. In this case the
Maxwellian function ${\displaystyle N_\eps (v)= \frac 1{2\eps} \e^{-\frac {|v|}{\eps}}}$ has finite
moments of any order, and the meaning of approximation is standard. Let us consider
equation \eqref{kin-heat} in weak form
\[
\frac {\d}{\d t} \int_\R g_\eps (v, t) \phi (v) \d v = \frac 1{\eps ^2}\iint_{\R^2}
\left(\phi(v+w) -\phi(v) \right ) g_\eps (v,t)N_{\eps} (w)  \d v \d w.
\]
By  expanding $\phi(v+w)$ in Taylor series around $v$
\[
\phi(v+w)= \phi(v) + \phi'(v) w + \frac {\phi''(v)}2 w^2 + \frac {\phi'''(\tilde v)}{3!} w^3, \quad \tilde v \in (v,w)
\]
we get
\[
\frac {\d}{\d t} \int_\R  g_\eps (v, t) \phi (v) \d v = \frac 1{\eps ^2}\iint_{\R^2}
\left(\phi'(v) { w} + \frac {\phi''(v)}2  {w^2} + \frac {\phi'''(\tilde v)}{3!}  {w^3}
\right ) g_\eps (v,t) {N_{\eps} (w)}  \d v \d w.
\]
Since
\[
{\int_\R  w \, N_\eps (w) \d w=0,\quad \int_\R  w^2 \, N_\eps (w) \d w =  2 \eps^2,
\quad \int_\R |w|^3 \, N_\eps (w) \d w = 12 \eps^3},
\]
we end up with the equation
\[
\frac {\d}{\d t} \int _\R g_\eps (v, t) \phi (v) \d v  = \int _\R g_\eps (v, t) \phi''
(v) \d v + C(\eps),
\]
where the remainder satisfies
\[
|C(\eps)| \leq 2 {\eps} \|\phi '''\|_{L^\infty} \int_\R|w|^3 N(w) \d w \to 0, \quad \eps \to 0.
\]
In the case of the fractional diffusion approximation, the operator
$(\sqrt{-\Delta})^\lambda$ is non-local and the Maxwellian $M_{\lambda,\eps}$ has finite moments
only for $\alpha <\lambda$. This requires a different way of looking to the problem.
We begin by defining in which sense we can consider equation \eqref{approx} as an
approximation of the fractional diffusion equation.

We recall that  a centered  distribution function $F$ belongs to
the domain of normal attraction of the stable law $L_\lambda(v) \d
v$ if for any sequence of independent and identically distributed
real-valued random variables $(X_n)_{n\geq 1}$ with common
distribution function $F$,  the law of
\be\label{somma}
\frac{X_1+\dots+X_n }{n^{1/\lambda}}
\ee
converges weakly to the stable law $L_\lambda(v) \d v$.

Let us recall some properties of a distribution $F$ belonging to the domain of normal attraction of a stable law. More information about this topic can be found in the book
\cite{Ibra71} or, among others, in the papers \cite{BLM}, \cite{BLR}.
 It is well-known that a
centered distribution $F$ belongs to the
 domain of normal attraction of the $\lambda$-stable law \eqref{levy} with density
$L_\lambda(v)$ if and only if $F$ satisfies $|v|^\lambda F(v)\to
c$ as $v\to -\infty$ and $v^\lambda (1-F(v))\to c$ as $v\to
+\infty$ i.e.
\begin{equation}\label{cardis}
\begin{aligned}
&F(-v)=\frac{c}{|v|^\lambda}+S_1(-v) \ \ \ \ \ \ {\rm and }\ \ \ \
\
\ \ 1-F(v)=\frac{c}{v^\lambda}+S_2(v) \ \ \ \ \ \ (v>0)\\
&S_i(v)=o(|v|^{-\lambda})\ \ \  \ {\rm as}\ |v|\to +\infty, \ \ \
i=1,2\\
\end{aligned}
\end{equation}
where $c=\frac{\Gamma(\lambda)}{\pi}\sin\left(\frac{\pi\lambda}{2}\right)$.
Concerning the moments of the distribution function $F$ and of the distributions $F_n$
associated to the sum $(X_1+\dots+X_n )/{n^{1/\lambda}}$ considered in \eqref{somma},
we recall the following Propositions.

\begin{prop}[see \cite{Ibra71}, Theorem 2.6.4 page 84]
Let $F$  belong to the domain of normal attraction of $L_\lambda$. Then, for any
$\alpha$ such that $0< \alpha<\lambda$
\[
\int_{\R}|v|^{\alpha} \d F(v)<+\infty.
\]
\end{prop}

\begin{prop}[see \cite{Ibra71}, Lemma 5.2.2 page 142]  \label{522}
Let $F_n$  denote the distribution function associated to the sum $(X_1+\dots+X_n)/
n^{1/\lambda}$ converging weakly to the stable law $L_\lambda(v)\d v$. Then, for any
$0<\alpha<\lambda$
\[
\int_{\R}|v|^{\alpha} \d F_n(v)
\]
is uniformly bounded with respect to $n$.
\end{prop}
The behavior of $F$ in the physical space \eqref{cardis} leads to a characterization
of the domain of normal attraction of $L_\lambda$ in terms of characteristic
functions. Indeed, if $\Phi$ is the characteristic function of the distribution
function $F$ satisfying \eqref{cardis} then
\[
1-\Phi(\xi)=\left(1+v_0(\xi)\right)|\xi|^\lambda,
\]
where
\[
v_0 \in L^\infty(\R)\quad   \text{  and }\quad  |v_0(\xi)|=o(1), \quad |\xi|\to 0.
\]
In the following we will consider a stronger assumption on the characteristic function.  We will denote by
 $M_\lambda  $ any   density of
a centered distribution function belonging to the domain of normal attraction of the
stable law  with density $L_\lambda$, which has the extra property that the
Fourier-Stieltjes transform of the function  $M_\lambda $ satisfies
\begin{equation}\label{Fourier1}
1-\hat{M_\lambda}(\xi)=\left(1+v_0(\xi)\right)|\xi|^\lambda,
\end{equation}
where $v_0(\xi)$ is such that, for some  $\delta >0$
  \be\label{ogrande}
   v_0 \in
L^\infty(\R)\quad \text{ and }\quad   |v_0(\xi)|=O\left(|\xi|^\delta\right), \quad
|\xi|\to 0.
 \ee
A main example is furnished by the so-called {\it Barenblatt} function
\be\label{bar}
B_\lambda(v)=\frac{\alpha}{(1+(\beta v)^2)^{{(1+\lambda) }/{2}}},\quad v\in\R
\ee
where $\alpha,\beta>0$, \quad $\displaystyle{\frac{\beta}{\alpha}=\int_{\R}\frac{\d
v}{(1+v^2)^{{(1+\lambda)}/2}}}$ and $\displaystyle{\frac{\alpha}{\lambda
\beta^{1+\lambda}}=\frac{\Gamma(\lambda)\sin(\frac{\pi\lambda}2)}{\pi}}.$ This type
of functions is mainly related to the study of nonlinear equations for fast diffusion,
given by
 \be\label{nl2}
             \partial_t g(v,t) =  \partial_{vv}^2 g^p(v,t),
 \ee
as $p<1$. The Barenblatt function \fer{bar} corresponds to the case $p = (\lambda
-1)/(\lambda +1)$. The function $t^{(\lambda +1)/(2\lambda)}B_\lambda(t^{(\lambda
+1)/(2\lambda)}x)$  represents the intermediate asymptotic profile of the nonlinear
diffusion \fer{nl2} \cite{Vaz}. The case $p<1$  appears when modelling diffusion in
metals and ceramic materials. In these materials, in fact, over a wide range of
temperatures, the diffusion coefficient can be approximated as $u^{ -\alpha}$, where
$0 < \alpha < 2$ \cite{Ros2}.

The distribution function $F(x)=\int_{-\infty}^{x}B_{\lambda}(v) \d v$ is such that
$$
\lim_{x\to +\infty}x^\lambda(1-F(x))=\frac{\alpha}{\lambda\beta^{1+\lambda}}=\frac{\Gamma(\lambda)\sin\left(\frac{\pi\lambda}{2}\right )}{\pi}=c.
$$
Moreover for any $0<\delta<2$ there exists  a constant $C>0$ such that for any $x>0$ we have
$$
\left|  x^{\lambda+\delta}   (1-F(x))-c    x^\delta         \right|<C.
$$
This in enough to guarantee that the  Fourier transform of $B_\lambda$  satisfies the
extra property \eqref{Fourier1}-\eqref{ogrande} (see \cite{BLR}). Property
\eqref{ogrande} has been already considered in kinetic theory. In particular,  it has
been used to determine  the rate of convergence to equilibrium for the dissipative Kac
model \cite{BLR}.

Under this condition on the Maxwellian function, we can prove convergence  of the
approximated solution \eqref{approx} to the solution to the fractional diffusion
equation \eqref{frac-eq} in the Fourier-based distance $d_s$, as  $\ve \to 0$.

\begin{theo}\label{conveps}
Let  $1<\lambda <2$, and let $g(t)$ and $g_\varepsilon(t)$ be  the solutions of the
fractional diffusion equation \eqref{fraceq} and, respectively, of the Rosenau
approximation \eqref{approx}, corresponding to the same initial probability density
$g_0$. Let us suppose moreover that the Maxwellian $M_{\lambda}$ in \eqref{approx}
satisfies the extra properties \eqref{Fourier1}- \eqref{ogrande}. Then for any
$0<s<\lambda$
 there exists a positive  constant $C=C(\lambda, s, \delta)$ such that
$$
d_s(g(t), g_\varepsilon(t)) \leq
C t^{s/(\lambda + \delta)} \eps^{s\delta /(\lambda + \delta)}.
$$

\end{theo}

\begin{proof}
Since $g_0$ is a probability density, $|\hat g_0(\xi)| \le 1$, and we have
\[
d_s(g(t),g_\varepsilon (t))=
\sup_{\xi\neq 0}
\frac{\left|\hat g_0(\xi)
\left(
{\rm e}^{-|\xi|^\lambda t}-{\rm e}^{-{\eps^{-\lambda}}t\left(1-\hat M_{\lambda,\eps}(\xi)\right)}
\right)
\right|}{|\xi|^s}
 \leq \sup_{\xi\neq 0}
\frac{
\left|  {\rm e}^{-|\xi|^\lambda t }-{\rm e}^{-\eps^{-\lambda}t\left(1-\hat M_{\lambda,\eps}(\xi)\right)}\right|}
{|\xi|^s}.
\]
Therefore, for any $R>0$
 \[
d_s(g(t),g_\varepsilon (t))\leq \frac{2}{R^s}+\sup_{0<|\xi|\leq R}
\frac{\left|  {\rm e}^{-|\xi|^\lambda t}-{\rm e}^{-{\eps^{-\lambda}}t \left(1-\widehat
M_{\lambda, \eps}(\xi )\right)}\right|} {|\xi|^s}.
\]
 Thanks to the inequality $|{\rm e}^{-x}-{\rm e}^{-y}|\leq |x-y|$, valid
for any $x,y\geq 0$, we get
\[
\sup_{0<|\xi|\leq R}\frac{\left||\xi|^\lambda
t-{\eps^{-\lambda}}{t}\left(1-\widehat M_{\lambda,\eps}(\xi)\right)\right|} {|\xi|^s}
\leq \sup_{0<|\xi|\leq R}\frac{{t|\xi|^\lambda} \left|v_0\left({\varepsilon\xi}\right)\right|}{{|\xi|^s}}.
 \]
In the last inequality we used the expression of the Fourier transform of $M_{\lambda}$ \eqref{Fourier1}.
Thanks to \eqref{ogrande} we obtain
 $$
 \sup_{0<|\xi|\leq R}\frac{{t|\xi|^\lambda} \left|v_0\left({\varepsilon\xi}\right)\right|}{{|\xi|^s}} \leq C \sup_{0<|\xi|\leq R}
\frac {t|\xi|^\lambda \left|\varepsilon \xi \right|^{\delta}}{{|\xi|^s}}\leq
C t R^{\lambda+\delta-s}\varepsilon^\delta.
$$
 Finally,  choose $\displaystyle{R= \left(\frac 2 {Ct\eps^\delta}\right )^{1/(\lambda +\delta)}}$ to obtain
$$
d_s(g(t),g_\varepsilon (t)) \leq
C t^{s/(\lambda + \delta)} \eps^{s\delta /(\lambda + \delta)}.
$$

\end{proof}


\section{Large time behavior}\label{main-sect}

The result of the previous section justifies the choice of a
Barenblatt type Maxwellian density in the linear kinetic model
\fer{Ros2}, to obtain an explicit in space linear kinetic equation
which approximates, at any fixed time, the fractional diffusion
equation. The result of Theorem \ref{conveps}, however, is such
that the rate of convergence in the $d_s$-metric is
time-dependent, and fails as $t \to \infty$. As the analysis in
\cite{RT} shows, the weakness of this result with respect to time
could be generated by the choice of a general Maxwellian in the
linear kinetic model, that, while maintaining the kinetic form of
the approximation, is loosing the precise shape of the Maxwellian
predicted by the Rosenau idea.  Consequently, in this section we
will restrict the study of the large time behavior of the solution
to the approximated Rosenau equation \eqref{Ros2}, where
$M_{\lambda,\ve}$ is a Linnik distribution \eqref{Max-f}.
{Provided that we discard the singular part of the approximating
solution, our analysis will confirm that in this case the solution
to \fer{Ros2} behaves  like the fractional diffusion equation for
large times.}

\subsection{Convergence in the Fourier based metric}
In analogy with  the solution of the fractional diffusion equation,  for all $\xi \in
\R$ we have
\[
\lim_{t \to +\infty} \left(\hat{g_\varepsilon}(t,\xi)-{\rm
e}^{-|\xi|^\lambda t} \right )= 0.
\]
This convergence can be refined using the Fourier-based distance
 $d_s$.
 In order to capture the asymptotic profile in the limit
for $t\to +\infty$ we  consider the scaled solution obtained by the change of variable
$\xi\longmapsto \frac{\xi}{(1+t)^{1/\lambda}}$. Let us denote by
 \be \label{riscalate}
\begin{aligned}
&h(v,t)=(1+t)^{1/\lambda}g((1+t)^{1/\lambda}v,t)\\
&h_\varepsilon(v,t)=(1+t)^{1/\lambda}g_\varepsilon((1+t)^{1/\lambda}v,t).
\end{aligned}
 \ee
 the scaled solutions of the fractional diffusion equation and of the approximated
equation respectively.
 Then using
the explicit representation of the solution established in
\eqref{p_epslambda-spazio}-\eqref{p_epslambda-fur}, we get
\[
\begin{aligned}
\hat{h}_\varepsilon(\xi,t)&=\hat{P}_{\lambda,\varepsilon}\left(\frac{\xi}{(1+t)^{1/\lambda}}, t\right)\hat
{g}_0\left(\frac{\xi}{(1+t)^{1/\lambda}}\right)\\
&= {\rm exp}\left(-{\varepsilon}^{-\lambda}t\left(1-\hat
{M}_{\lambda,\varepsilon}\left(\frac{\xi}{(1+t)^{1/\lambda}}\right)\right)
\right)\hat{g}_0\left(\frac{\xi}{(1+t)^{1/\lambda}}\right)
\end{aligned}
\]
and  since
\[
{\varepsilon}^{-\lambda}t\left(1-\hat
{M}_{\lambda,\varepsilon}\left(\frac{\xi}{(1+t)^{1/\lambda}}\right)\right)=\frac{t|\xi|^\lambda}{1+t+\varepsilon^{\lambda}|\xi|^\lambda}
\]
for $t\to +\infty$ we get, for $\xi \in \R$ and $\eps >0$ fixed
$$
\lim_{t \to +\infty} \hat{h}_\varepsilon(\xi,t) ={\rm
exp}\left({-|\xi|^\lambda }\right){\hat g}_0(0).
$$
\begin{prop} \label{dscon}
Let  $1<\lambda <2$, and let $h(t)$ and $h_\varepsilon(t)$ be  the solutions of
\eqref{riscalate}, the scaled fractional diffusion equation \eqref{fraceq} and,
respectively, of its Rosenau approximation \eqref{Ros2}, corresponding to the same
initial probability density $g_0$. If $d_s(g_0,L_\lambda)<+\infty$ for some
$0<s<\lambda$, then  there exists a positive constant $C = C(\lambda, s)$ such that
\be\label{prima}
d_s(h(t), h_\varepsilon(t))\leq C\ \varepsilon
^{s/2}\frac{t^{s/2\lambda}}{(1+t)^{s/\lambda}},
\ee
and
\be\label{seconda}
d_s(h_\varepsilon(t), L_\lambda)\leq C \left(\frac {1}{(1+t)^{s/\lambda}}+
\varepsilon^{s/2}\frac{t^{s/2\lambda}}{(1+t)^{s/\lambda}}\right).
\ee
\end{prop}

\begin{proof}
By the scaling rule \eqref{scala}, we have
\[
d_s(h(t), h_\varepsilon(t))= \frac 1{(1+t)^{s/\lambda}} d_s(g(t), g_\eps(t)).
\]
The first inequality \eqref{prima} follows therefore from Theorem \ref{conveps}. Indeed, ${\displaystyle \widehat M_\lambda (\xi) =  \frac 1{1+|\xi|^\lambda}}$ and so
\[
1-\widehat M_{\lambda}(\xi)= |\xi|^\lambda \left(1- \frac {|\xi|^\lambda}{1+|\xi|^\lambda}\right ) = |\xi|^\lambda \left(1+v_0(\xi)\right )
\]
with
\[
v_0(\xi)= - \frac {|\xi|^\lambda}{1+|\xi|^\lambda} = O(|\xi|^\lambda), \quad | \xi| \to 0.
\]

The second inequality \eqref{seconda} is a consequence of the first one. Indeed, we apply  the triangular
inequality
 \[
d_s(h_\varepsilon(t), L_\lambda)\leq  d_s\left(h_\varepsilon(t),
h(t)\right)+d_s\left(h(t),L_\lambda\right),
 \]
where
\[
\begin{aligned}
d_s\left(h(t),L_\lambda\right)=&\sup_{\xi\neq 0}\frac{\left|\hat
g_0\left(\frac{\xi}{(1+t)^{1/\lambda}}\right){\rm
exp}(-\frac{|\xi|^\lambda t}{1+t})-{\rm exp}(-|\xi|^\lambda
)\right|}{|\xi|^{s}}
\leq \frac{1}{(1+t)^{s/\lambda}}d_s(g_0,L_\lambda).
\end{aligned}
 \]
\end{proof}

\subsection{Strong convergence in $L^1$ of a regularized solution}

As given by expression \fer{fu2}, the fundamental solution $P_{\lambda,\eps}(\cdot
,t)$ contains a singular part, with a size which is exponentially decaying to zero.
The effect of this singular part on the solution to the appro\-xi\-mation is clear. For
this reason, we study here the large time behavior of the regularized part. The main
result will be that this suitable regularized solution to the Rosenau approximation
converges in strong sense (namely in $L^1$-norm) to the solution of the fractional
diffusion equation.

We define the regularized fundamental solution to the Rosenau approximation by
 \be\label{rego1}
 g_{\varepsilon,reg}(v,t)=P_{\lambda,\eps,
reg}(\cdot , t)\ast g_0(v)
 \ee
where
\begin{equation}\label{nucleoreg}
P_{\lambda, \eps,reg}(v,t)=P_{\lambda,\eps}(v,t)+{\rm e}^{-\ve^{-\lambda} t}\left(  M_{\lambda,\eps}(v) - \delta_0(v)\right)
\end{equation}

\begin{rem}
Note that function \fer{nucleoreg} is obtained from \fer{fu2} by substituting the
singular part with the order zero term of the sum. In this way, the difference between
the fundamental solution and the regularized one is still exponentially decaying to
zero both with respect to time and $\ve$, with the additional property to vanish at
the point $\xi= 0$. In other words, the regularized solution is constructed in such a
way that the masses of $P_{\lambda,\eps}(\cdot,t)$ and $P_{\lambda,
\eps,reg}(\cdot,t)$ coincide.
\end{rem}

\begin{rem}
It is immediate to prove that  $P_{\lambda, \eps,reg}(t)\in L^1(\R)$ for each value of
$t>0$. Indeed
 \be\label{espressione} P_{\lambda,\eps,reg}(t)={\rm e}^{-\ve^{-\lambda}
t}\sum_{n=1}^{+\infty}\left({t\eps^{-\lambda}}\right)^n\frac 1{n!}\
M^{*n}_{\lambda,\eps}+{\rm e}^{-\ve^{-\lambda} t}M_{\lambda,\eps},
 \ee
 and the
Maxwellian term $M_{\lambda,\eps}$ belongs to  $L^1(\R)$. Then the series in \fer{espressione}
converges in $L^1(\R)$ for any $t>0$.
\end{rem}
 The main result of this paper is contained in the following

\begin{theo}\label{main}
Let $1<\lambda <2$, and let $g(t)$ and $ g_{\varepsilon,reg}(t)$ denote  the solutions
of the Cauchy problem for the fractional diffusion equation \eqref{fraceq} and,
respectively,  the  solution of the regularized Rosenau approximation \fer{rego1},
corresponding to the same initial density $g_0$. Then, if for all $0<\alpha<\lambda$,
$g_0$ has bounded moment of order $\alpha$
$$
\lim_{t\to +\infty}\left\|g_{\eps,reg}(t)-g(t)\right\|_{L^1}=0
$$
\end{theo}
By Theorem \ref{main} and Proposition \ref{prop1} we get immediately the following Corollary.
\begin{coro}
Let $1<\lambda <2$, and let $g_0$ be  a probability density with bounded moment of
order $\alpha$,  for all $0<\alpha<\lambda$. Then
$$
\lim_{t\to +\infty}\left\|g_{\eps,reg}(t)-\frac 1{t^{1/\lambda}}L_\lambda \left( \frac \cdot {t^{1/\lambda}}\right )\right\|_{L^1}=0
$$
with $L_\lambda$ the L\'evy symmetric stable distribution \eqref{levy}.
\end{coro}
The proof of Theorem  \ref{main} follows by various steps, we will split into different
lemmas, that we will prove below. The first one is concerned with the convergence
(after scaling) of the fundamental solution of the fractional diffusion equation to
the regularized fundamental solution.
For the sake of simplicity let us denote
\begin{equation}\label{fundresc}
\begin{aligned}
& \tilde P_{\lambda}(v,t)=(1+t)^{1/\lambda}P_{\lambda}((1+t)^{1/\lambda}v,
t),\\
& \tilde P_{\lambda, \eps, reg}(v,t)=(1+t)^{1/\lambda}P_{\lambda,\varepsilon, reg}((1+t)^{1/\lambda}v,
t).
\end{aligned}
\end{equation}
The
result follows in consequence of Proposition \ref{dscon}.
\begin{lem}\label{lemma3}
Let $1<\lambda<2$. For any $0<s<\lambda$ there exists a positive constant $C =
C(\lambda,s)$  such that
$$
d_s(\tilde P_\lambda (t),\tilde P_{\lambda,\eps,reg}(t)) \leq C
\frac{\eps^{ s/2}}{(1+t)^{s/(2\lambda)}}, \quad t>0.
$$
\end{lem}
The second lemma describes the growth of the $\alpha$-moment of
the scaled fundamental solution of the fractional diffusion
equation and of the scaled regularized fundamental solution of the
Rosenau approximation.
\begin{lem}\label{lemma4}
Let $1<\lambda <2$. For any $0<\alpha<\lambda$ and any $t > 0$
\[
m_\alpha(\tilde P_\lambda(t)) = \int_\R |v|^\alpha  \tilde P_\lambda(v, t)\,  \d v
\leq m_\alpha(L_\lambda).
 \]
Moreover, there exist positive constants $C=C(\lambda,\alpha)$ such that
 \[
 m_\alpha(\tilde P_{\lambda, \eps,reg}(t))= \int_\R
|v|^\alpha \tilde P_{\lambda, \eps,reg}(v, t)\, \d v  \leq  C.
 \]
\end{lem}
The third lemma deals with the evolution of the Sobolev norms of both the scaled
fundamental solution of the fractional diffusion equation and of the scaled
regularized fundamental solution of the Rosenau approximation.
\begin{lem}\label{lemma5}
Let $1<\lambda <2$. For any $0<s<(\lambda-1)/2$, any $0<\beta<{1
/\lambda}$ and for $t$ large enough  there are positive constants
$C_1=C_1(\lambda,s)$ and $C_2= C_2(\lambda,\beta,s)$ such that
\[
\begin{aligned}
&\left\| \tilde P_\lambda (t)
\right \|_{\dot H^s} \leq C_1, \\
& \left\| \tilde P_{\lambda, \eps, reg}(t) \right \|_{\dot H^s}
\leq C_2\frac 1 {\eps^{s+{1 /2}}}(1+t)^{(s+
1/2)\left({1 /\lambda} -\beta\right )}.
\end{aligned}
\]
\end{lem}

\bigskip

With these results at hand, we can prove Theorem \ref{main}.

\noindent {\bf Proof of Theorem \ref{main}.} The proof of the theorem can be
divided into different steps. Recalling that $\|g_0\|_{L^1} =1$, $g(t)=P_\lambda(t)
\ast g_0$ and $g_{\eps,reg}(t)=P_{\lambda, \eps,reg}(t)\ast g_0$, we get for $t > 0$
$$
\left\| g(t)-g_{\eps,reg}(t)\right\|_{L^1}\leq \left\|
P_\lambda(t)-P_{\lambda, \eps,reg}(t)\right\|_{L^1} .
$$
Thanks to the invariance by scaling of the $L^1$ norm we get
$$
\left\| P_\lambda (t)-P_{\lambda,
\eps,reg}(t)\right\|_{L^1}=\left\| \tilde P_\lambda(t)-\tilde
P_{\lambda, \eps,reg}(t)\right\|_{L^1},
$$
with $\tilde P_\lambda$ and $\tilde P_{\lambda, \eps,reg}$ defined in
\eqref{fundresc}. By using the  interpolation inequalities  \fer{ine!}, for $t>0$  and
$\alpha, s \in (0,\lambda)$ we obtain
 \be \label{bound}
\begin{aligned}
\left\| \tilde P_\lambda(t)-\tilde P_{\lambda, \eps,reg}(t)
\right\|_{L^1} &\leq C \, \left\| \tilde P_\lambda(t)-\tilde
P_{\lambda, \eps,reg}(t)\right\|_{L^2}^{\frac
{2\alpha}{1+2\alpha}}
\left[m_\alpha(\tilde P_\lambda(t))+m_\alpha(\tilde P_{\lambda, \eps,reg}(t))\right]^{\frac 1{1+2\alpha}}\\
& \leq C\, d_s\left(\tilde P_\lambda (t),\tilde
P_{\lambda, \eps,reg}(t)\right)^{\frac{4s\alpha}{(1+4s)(1+2\alpha)}}\left\|\tilde
P_\lambda (t)-\tilde
P_{\lambda, \eps,reg}(t)\right\|_{\dot H^s}^{\frac{2\alpha(1+2s)}{(1+4s)(1+2\alpha)}}\\
&\quad\quad\quad\quad\quad\quad
 \left[m_\alpha(\tilde P_\lambda(t))+m_\alpha(\tilde P_{\lambda, \eps,reg}(t))\right]^{\frac 1{1+2\alpha}}.\\
\end{aligned}
\ee
 Thanks to estimate \eqref{bound}, the bounds in Lemmas \ref{lemma3}, \ref{lemma4} and \ref{lemma5}
 imply
\[
\begin{aligned}
&\| \tilde P_\lambda (t) -\tilde P_{\lambda, \eps, reg}(t)\|_{L^1} \\
&\leq C \left[ \frac {\eps^{\frac s2}} {(1+t)^{\frac s
{2\lambda}}} \right ]^{\frac {4s\alpha}{(1+4s)(1+2\alpha)}} \left[
m_\alpha (L_\lambda) +C \right ]^{\frac 1{1+2\alpha}} \left[ C_1
+C_2
 \frac 1{\eps^{s+{1 /2}}}(1+t)^{(s+{1 /2})({1 /\lambda} -\beta)}\right ]^{\frac {2\alpha(1+2s)}{(1+4s)(1+2\alpha)}}\\
&\leq C \eps^{\gamma(s,\alpha)} (1+t)^{\delta(s,\alpha)}
\end{aligned}
\]
where $C>0$ depends on $\lambda$, $s$ and $\beta$, $\gamma(s,\alpha)=\frac{4s^2\alpha-
2\alpha(1+2s)^2}{2(1+4s)(1+2\alpha)}$ and  $\delta(s,\alpha)=\frac {-4s^2\alpha
+2\alpha(1+2s)^2(1-\lambda \beta)}{2\lambda(1+4s)(1+2\alpha)}$. By choosing $1-\lambda
\beta$ small enough, we get  $\delta(s,\alpha)<0$. Hence the result follows.

\subsection{Proofs of the Lemmas}

{\bf Proof of Lemma \ref{lemma3}.} We have
\[
\begin{aligned}
& d_s(\tilde P_\lambda (t),\tilde
P_{\lambda,\eps,reg}(t))\\
&=\sup_{\xi\neq 0}
\frac 1{|\xi|^s}
{\left|
{\rm e}^{-|\xi|^\lambda \frac t{1+t}}-
\left[
{\rm e}^{-\eps^{-\lambda}{t}\left(1-\widehat M_{\lambda,\eps}\left(\frac \xi {(1+t)^{1/ \lambda}}\right)\right)}
-{\rm e}^{-\ve^{-\lambda} t }\left(1-\widehat M_{\lambda,\eps}
\left(\frac \xi {(1+t)^{{1/\lambda}}}\right)\right)
\right]
\right|}\\
 & \leq \sup_{\xi\neq 0}\frac 1{|\xi|^s}{\left|  {\rm
e}^{-|\xi|^\lambda \frac t{1+t}}-{\rm e}^{-\eps^{-\lambda}{t}\left(1-\widehat
M_{\lambda,\eps}\left(\frac \xi {(1+t)^{{1/\lambda}}}\right)\right)}\right|}+\sup_{\xi\neq
0}\frac 1{|\xi|^s}{\left|{\rm e}^{-\ve^{-\lambda} t}\left(1-\widehat
M_{\lambda,\eps}
\left(\frac {\xi}{(1+t)^{{1/\lambda}}}\right)\right)\right|} \\
&=I+II.
\end{aligned}
\]
The  term $II$ can be written in the form
\[
 II=\sup_{\xi\neq
0}\frac{\left|{\rm e}^{-\ve^{-\lambda} t}\left(1-\widehat
M_{\lambda,\eps}
\left(\frac{\xi}{(1+t)^{{1/\lambda}}}\right)\right)\right|}{|\xi|^s}=\frac{{\rm
e}^{-\ve^{-\lambda} t}}{(1+t)^{{s /\lambda}}}\sup_{\xi\neq
0}\frac{\left|1-\widehat M_{\lambda,\eps}(\xi)\right|}{|\xi|^s}.
 \]
Therefore, since $1 < \lambda < 2$,  and  $0<s<\lambda$  for $\xi \neq 0$ we get
 \[
 \begin{aligned}
\frac{\left|1-\widehat M_{\lambda,\eps}(\xi)\right|}{|\xi|^s} &= \frac {
\eps^\lambda|\xi|^\lambda}{|\xi|^s (1+\eps^\lambda |\xi|^\lambda)} \leq \eps^\lambda
|\xi|^{\lambda -s}.
\end{aligned}
\]
Thus,  for any given positive constant $R>0$
\[
\sup_{\xi\neq 0}\frac{\left|1-\widehat M_{\lambda,\eps}(\xi)\right|}{|\xi|^s}
\leq \frac 1{R^s} + \sup_{0<|\xi|\leq R}\frac{\left|1-\widehat M_{\lambda,\eps}(\xi)\right|}{|\xi|^s} \leq \frac 1{R^s} + \eps^\lambda R^{\lambda -s}.
\]
By choosing $R=1/\eps$ we get
\[
II\leq 2 \eps^s.
\]
Finally, we get the bound
\[
 II\leq C \eps^s\frac {\e^{-\ve^{-\lambda} t}}{(1+t)^{{s /\lambda}}}.
\]
The first term coincides with the term estimated in Proposition \ref{dscon}. For this
term  we proved the bound
$$
I\leq C
\varepsilon^{s/2}\frac{t^{s/{(2\lambda)}}}{(1+t)^{s/\lambda}}.
$$
This is enough to conclude.
\bigskip

Let us now estimate the moments.

\noindent {\bf Proof of Lemma \ref{lemma4}.} Let us remind that
\[
\tilde P_\lambda(v,t) =(1+t)^{{1 /\lambda}}
P_\lambda((1+t)^{{1 /\lambda}}v,t)
\]
and that  in \eqref{mom-plambda} we stated that
\[
\int_\R |v|^\alpha   P_\lambda(v, t) \d v=   t^{ \alpha/\lambda} m_\alpha(L_\lambda).
\]
Hence we obtain
\[
\begin{aligned}
\int_\R |v|^\alpha   \tilde P_\lambda(v, t) \d v &=   \int_\R |v|^\alpha  (1+t)^{{1 /\lambda}}
P_\lambda((1+t)^{{1 /\lambda}}v,t) \, \d v= \frac 1{(1+t)^{{\alpha /\lambda}}} \int_\R   |v|^\alpha   P_\lambda(v, t) \d v \leq m_\alpha(L_\lambda).
\end{aligned}
\]
To estimate $m_\alpha(\tilde P_{\lambda, \eps,reg}(t))$, we remark that, thanks to
expression \eqref{espressione}, it is enough to prove that
\[
\int_\R |v|^\alpha \left( {\rm e}^{-\ve^{-\lambda} t}\sum_{n=1}^{+\infty}\left(\frac
t{\eps^\lambda}\right)^n\frac 1{n!}\
M^{*n}_{\lambda,\eps} (v)\right ) \, \d v \leq  t^{{\alpha /\lambda}}C.
\]
By definition
\[
 M^{*n}_{\lambda,\eps} (v) = \eps^{-1}  M^{*n}_{\lambda} (\eps^{-1}v),
 \]
and this implies
 \[
 \int_\R |v|^\alpha M^{*n}_{\lambda,\eps} (v)\, \d v = \eps^\alpha \int_\R |v|^\alpha M^{*n}_{\lambda} (v)\, \d v.
\]
Now, consider that $n^{1/\lambda}M^{*n}_{\lambda} (n^{1/\lambda}v)$ is the density of
 the sum ${(X_1+\dots+X_n )/{n^{1/\lambda}}}$, where
$(X_n)_{n\geq 1}$ is a
 sequence of
independent and identically distributed real-valued random variables with common
density $M_\lambda$ in the domain of normal attraction of $L_\lambda (v)\d v$. This
fact can be used to write
\[
  \int_\R |v|^\alpha M^{*n}_{\lambda} (v)\, \d v=   n^{\alpha/\lambda} \int_\R |w|^\alpha n^{1/\lambda} M^{*n}_{\lambda} (n^{{1 /\lambda}}w) \, \d w.
 \]
Since $M_\lambda$ is a centered distribution, ${(X_1+\dots+X_n )/{n^{{1 /\lambda}}}}$
converges in law to $L_\lambda(v) \d v$. Therefore,   thanks to Lemma 5.2.2 in
\cite{Ibra71} (see Proposition \ref{522}) one obtains
\[
\int_\R |w|^\alpha n^{{1 /\lambda}} M^{*n}_{\lambda} (n^{{1 /\lambda}}w) \, \d w
\leq C,
\]
and this implies
\[
 \int_\R |v|^\alpha
M^{*n}_{\lambda,\eps} (v) \, \d v \leq C\eps ^\alpha
n^{ \alpha/\lambda} .
\]
In conclusion
\[
\begin{aligned}
\int_\R |v|^\alpha \left( {\rm e}^{-\ve^{-\lambda} t}\sum_{n=1}^{+\infty}\left(\frac
t{\eps^\lambda}\right)^n\frac 1{n!}\
M^{*n}_{\lambda,\eps} (v)\right ) \, \d v &=     {\rm e}^{-\ve^{-\lambda} t}\sum_{n=1}^{+\infty}\left(\frac
t{\eps^\lambda}\right)^n\frac 1{n!}
\int_\R |v|^\alpha\
M^{*n}_{\lambda,\eps} (v) \, \d v \\
&\leq   C \eps ^\alpha \ {\rm e}^{-\ve^{-\lambda} t}\sum_{n=1}^{+\infty}\left(\frac
t{\eps^\lambda}\right)^n\frac 1{n!}
n^{{\alpha /\lambda}} .
\end{aligned}
\]
We will now prove that
\[
{\rm e}^{-\ve^{-\lambda} t}\sum_{n=1}^{+\infty}\left(\frac
t{\eps^\lambda}\right)^n\frac 1{n!} n^{{\alpha /\lambda}}  \leq \frac {t^{\alpha /\lambda}}{\eps ^\alpha},
\]
or, in an equivalent way
\[
\sum_{n=1}^{+\infty}\left(\frac
t{\eps^\lambda}\right)^n\frac 1{n!}
\left(\frac {n}{t}\right )^{{\alpha /\lambda}}  \leq \frac {{\rm e}^{\ve^{-\lambda} t}}{\eps^\alpha}.
\]
If  $ \tau= \ve^{-\lambda} t, $ and $\beta= {\alpha /\lambda}$, this amounts to prove
that, if $0<\beta<1$
\[
\sum_{n=1}^{+\infty}\frac {\tau^n}{n!} \left(\frac {n}{\tau}\right )^{\beta}  \leq
{\rm e}^{\tau}.
\]
To this end, it is enough to remark that, if $\beta=0$
\[
\sum_{n=1}^{+\infty}\frac {\tau^n}{n!}  =  {\rm e}^{\tau}-1.
\]
On the other side, if  $\beta=1$ we get
\[
\sum_{n=1}^{+\infty}\frac {\tau^n}{n!} \left(\frac {n}{\tau}\right )=
\sum_{n=0}^{+\infty}\frac {\tau^n}{n!} = {\rm e}^{\tau}.
\]
Hence, if $0< \beta <1$
\[
\sum_{n=1}^{+\infty}\frac {\tau^n}{n!}
\left(\frac {n}{\tau}\right )^{\beta}  \leq \sum_{n=1}^{+\infty}\frac {\tau^n}{n!}
\max \left(1, \frac {n}{\tau}\right ) \leq \max \left(  \sum_{n=1}^{+\infty}\frac {\tau^n}{n!} ,  \sum_{n=0}^{+\infty}\frac {\tau^n}{n!}     \right )  = {\rm e}^{\tau}.
\]
All these bounds imply
\[
 \int_\R |v|^\alpha  \tilde P_{\lambda, \eps,reg}(v, t)\, \d v  \leq \frac{t^{{\alpha /\lambda}}}{(1+t)^{{\alpha /\lambda}}}C  \leq C.
\]

\begin{rem}\label{rem7} A direct consequence of this Lemma is that the moment
of order $\alpha$ of the approximated solution $g_\varepsilon (t)$ of equation
\eqref{approx}, grows at a polynomial rate in time in correspondence to any  generic
centered distribution $M_\lambda (v)\d v$ belonging to the domain of normal attraction
of the stable law $L_\lambda(v) \d v$
\[
\int_\R  |v|^\alpha g_\eps(v,t) \d v \leq C t^{{\alpha /\lambda}}.
\]
\end{rem}


\noindent{\bf Proof of Lemma \ref{lemma5}.} By definition
\[
\tilde P_\lambda (v,t) = \left( \frac {1+t}{t}\right )^{{1
/\lambda}}L_\lambda \left( \left(\frac {1+t}{t}\right )^{{1
/\lambda}}v\right ).
\]
Therefore, for $t >0$ large enough and $s>0$
\[
 \left\|
\tilde P_\lambda (t) \right \|_{\dot H^s} = \left( \frac
{1+t}{t}\right )^{\left( s+{1 /2}\right )/\lambda } \left\|
L_\lambda\right \|_{\dot H^s} \leq C_1.
\]
Expressions \eqref{nucleoreg} and \eqref{p_epslambda-fur} lead to
\[
{\cal F}\left(\tilde P_{\lambda, \eps, reg}\right )(\xi, t) =
\e^{-\ve^{-\lambda} t} \left[ \exp\left( \ve^{-\lambda}\  t
\widehat M_{\lambda, \eps} \left(
 \frac \xi{(1+t)^{{1 /\lambda}}}
 \right )
\right ) -1+\widehat M_{\eps,\lambda} \left( \frac \xi{(1+t)^{{1 /\lambda}}} \right
) \right ].
\]
Therefore
  \be\label{sob}
\begin{aligned}
\left\| \tilde P_{\lambda, \eps, reg}( t) \right\|_{\dot
H^s}^2\leq & C \int_\R |\xi|^{2s}
\exp\left(-{2t\eps^{-\lambda}}\right ) \left( \exp\left(
\ve^{-\lambda} t\ \widehat M_{\lambda, \eps} \left( \frac
\xi{(1+t)^{ 1 /\lambda}} \right ) \right ) -1
\right )^2\, \d \xi +\\
&  \int_\R |\xi|^{2s} \exp\left(-{2t\eps^{-\lambda}}\right )
\widehat M_{\lambda, \eps}
\left(
\frac \xi{(1+t)^{{1 /\lambda}}}
\right )^2\, \d \xi = A+B.
\end{aligned}
\ee
 By a change of variable, term  $B$ can be estimated as follows
\be\label{stima4}
\begin{aligned}
B &= (1+t)^{{(2s+1)/\lambda}}\exp\left(-{2t\eps^{-\lambda}}\right
) \left\| M_{\lambda,\eps}\right \|_{\dot H^s}^2 =
(1+t)^{{(2s+1)/\lambda}}\frac{\exp\left(-{2t\eps^{-\lambda}}\right
)}{\eps^{2s+1}} \left\| M_{\lambda}\right \|_{\dot H^s}^2\\
&\leq C_{\lambda, s}
(1+t)^{{(2s+1)/\lambda}}\frac{\exp\left(-{2t\eps^{-\lambda}}\right
)}{\eps^{2s+1}}
\end{aligned}
  \ee
  for $s< \lambda -1/2$.

Likewise, the term $A$ can be expressed as
\[
A=  (1+t)^{{(2s+1)/\lambda}}
\int_\R | \eta |^{2s}
\left(
\exp\left(
-{t\eps^{-\lambda}} \left(1-\widehat M_{\lambda, \eps} (\eta)\right )
\right )
-\exp\left(-{t\eps^{-\lambda}}\right )
\right )^2\, \d \eta = (1+t)^{{(2s+1)/\lambda}}   I_\eps(t).
\]
To estimate the term $I_\eps(t)$, consider that it satisfies the differential equation
 \[
\frac {\d I_\eps(t)}{\d t} = -{2\eps^{-\lambda}} I_\eps(t) + {2\eps^{-\lambda}}
A_\eps(t),
 \]
where $A_\eps(t)$ can bounded in a precise way. In fact
 \[
\begin{aligned}
&\frac {\d I_\eps(t)}{\d t}
=
\int_\R | \eta |^{2s}  2
\left(
\exp\left(
-{t\eps^{-\lambda}} \left(1-\widehat M_{\lambda, \eps} (\eta)\right )
\right )
-\exp\left(-{t\eps^{-\lambda}}\right )
\right )\times\\
&  \qquad \qquad \qquad\left(
-\exp\left(
-{t\eps^{-\lambda}} \left(1-\widehat M_{\lambda, \eps} (\eta)\right )
\right )
\frac{1-\widehat M_{\lambda, \eps} (\eta)}{\eps^\lambda}  + \frac 1{\eps^\lambda}\exp\left(-{t\eps^{-\lambda}}\right )
\right )
\, \d \eta\\
=& -{2\eps^{-\lambda}}
\int_\R | \eta |^{2s}
\left(
\exp\left(
-{t\eps^{-\lambda}} \left(1-\widehat M_{\lambda, \eps} (\eta)\right )
\right )
-\exp\left(-{t\eps^{-\lambda}}\right )
\right )^2 \, \d \eta + \\
& {2\eps^{-\lambda}}
\int_\R | \eta |^{2s}
\left(
\exp\left(
-{t\eps^{-\lambda}} \left(1-\widehat M_{\lambda, \eps} (\eta)\right )
\right )
-\exp\left(-{t\eps^{-\lambda}}\right )
\right )
\widehat M_{\lambda, \eps} (\eta)\
\exp\left(-{t\eps^{-\lambda}}
 \left(1-\widehat M_{\lambda, \eps} (\eta)\right )
\right )
\, \d \eta\\
& = -{2\eps^{-\lambda}} I_\eps(t) + {2\eps^{-\lambda}} A_\eps(t).
\end{aligned}
\]
It remains to estimate the  term $A_\eps(t)$. For a given $\beta >0$, we split the
$A_\eps(t)$ term into two terms
\[
\begin{aligned}
A_\eps(t) = &\int_\R | \eta |^{2s} \widehat M_{\lambda, \eps} (\eta)
\exp\left(-{2t\eps^{-\lambda}}
 \left(1-\widehat M_{\lambda, \eps} (\eta)\right )\right )
\left(
1-
\exp\left(
-{t\eps^{-\lambda}} \widehat M_{\lambda, \eps} (\eta)
\right )
\right )
\, \d \eta \\
= &
\int_{0\leq |\eta| \leq \frac 1{(1+t)^\beta}}
| \eta |^{2s} \widehat M_{\lambda, \eps} (\eta)
\exp\left(-{2t\eps^{-\lambda}}
 \left(1-\widehat M_{\lambda, \eps} (\eta)\right )\right )
\left(
1-
\exp\left(
-{t\eps^{-\lambda}} \widehat M_{\lambda, \eps} (\eta)
\right )
\right )
\, \d \eta +\\
& \int_{{|\eta| > \frac 1{(1+t)^\beta}}}
 | \eta |^{2s} \widehat M_{\lambda, \eps} (\eta)
\exp\left(-{2t\eps^{-\lambda}}
 \left(1-\widehat M_{\lambda, \eps} (\eta)\right )\right )
\left(
1-
\exp\left(
-{t\eps^{-\lambda}} \widehat M_{\lambda, \eps} (\eta)
\right )
\right )
\, \d \eta \\
= & A_1+ A_2.
\end{aligned}
\]
Since $0< \widehat M_{\lambda, \eps}(\eta) \leq 1$, we have
\[
A_1 \leq   \int_{0\leq |\eta| \leq \frac 1{(1+t)^\beta}} | \eta |^{2s} \, \d \eta  =
\frac 2{2s+1}\ \frac 1{(1+t)^{\beta(2s+1)}} = \frac {C_s}{(1+t)^{\beta(2s+1)}}.
\]
Moreover
\[
\begin{aligned}
A_2&= \int_{{|\eta| > \frac 1{(1+t)^\beta}}}
 | \eta |^{2s} \widehat M_{\lambda, \eps} (\eta)
\exp\left(-{2t\eps^{-\lambda}}
 \left(1-\widehat M_{\lambda, \eps} (\eta)\right )\right )
\left(
1-
\exp\left(
-{t\eps^{-\lambda}} \widehat M_{\lambda, \eps} (\eta)
\right )
\right )
\, \d \eta
\\
& \leq  \sup_{|\eta| > \frac 1{(1+t)^\beta}}
 \exp\left(-{2t\eps^{-\lambda}}  \left(1-\widehat M_{\lambda, \eps} (\eta)\right )\right )
 \int_{{|\eta| > \frac 1{(1+t)^\beta}}}  | \eta |^{2s} \widehat M_{\lambda, \eps} (\eta)\, \d \eta\\
 &
 = C_{\lambda,s}  \frac 1{\eps^{2s+1}} \sup_{|\eta| > \frac 1{(1+t)^\beta}}
 \exp\left(-{2t\eps^{-\lambda}}  \left(1-\widehat M_{\lambda, \eps} (\eta)\right )\right ),
 \end{aligned}
\]
where $C_{\lambda,s}= \int_\R| \eta|^{2s} \widehat M_\lambda
(\eta)\, \d \eta <+\infty$ for $s<(\lambda -1)/2$. Now, since
$\displaystyle{\widehat M_{\lambda, \eps} (\eta) =\frac 1{1+|\eps
\xi|^\lambda}}$  we obtain
\[
\sup_{|\eta| > \frac 1{(1+t)^\beta}}
 \exp\left(-{2t\eps^{-\lambda}}
 \left(1-\widehat M_{\lambda, \eps} (\eta)\right )\right )
  = \sup_{|\eta| > \frac 1{(1+t)^\beta}}
\exp\left(-\frac {2t |\eta|^\lambda}{1+\eps^\lambda |\eta| ^\lambda}\right ).
  \]
Since the function $\exp\left(-\frac {2t |\eta|^\lambda}{1+\eps^\lambda |\eta| ^\lambda}\right ) $ is decreasing in $\eta$, we get
 \[
 \sup_{|\eta| > \frac 1{(1+t)^\beta}}
\exp\left(-\frac {2t |\eta|^\lambda}{1+\eps^\lambda |\eta |^\lambda}\right )
 = \exp\left(- 2t\, \frac { \frac 1{(1+t)^{\beta\lambda}}}{1+\eps^\lambda  \frac 1{(1+t)^{\beta\lambda}}}\right )= \exp\left(- \frac{2t}{\eps^\lambda + (1+t)^{\beta\lambda}}\right ).
    \]
If $0<\beta<{1 /\lambda}$, we conclude that there exists a  constant $C_\beta >0$
such that
\[
 \exp\left(- \frac{2t}{\eps^\lambda + (1+t)^{\beta\lambda}}\right) \leq \e^{-C_\beta t^{1-\beta\lambda}}.
 \]
In the end, for $t > 0$ and $0<\beta<{1 /\lambda}$ there exists a constant
$C= C(\lambda,s,\beta)>0$ such that
 \be\label{A(t)} A_\eps(t) \leq  \left( \frac {C_s}{(1+t)^{\beta(2s+1)}} +
\frac{C_{\lambda, s}}{\eps^{2s+1}}\e^{-C_\beta t^{1-\beta\lambda}}\right ) \leq   \frac
{C} {\eps^{2s+1}}\frac {1}{(1+t)^{\beta(2s+1)}}.
 \ee
 With this estimate we can study the differential relation
\[
\frac{\d I_\eps(t)}{\d t}= -{2\eps^{-\lambda}} I_\eps(t) + {2\eps^{-\lambda}} A_\eps(t).
\]
For $a>0$ to be suitably chosen, we get
\[
I_\eps(t) \e^{2{\eps^{-\lambda}} t} -
I_\eps(a\eps^\lambda)\e^{2a} = {2\eps^{-\lambda}} \int_{a\eps^\lambda}^t  A_\eps(\sigma) \e^{{2\eps^{-\lambda}}\sigma}\, \d \sigma.
\]
Since, for $0<s<\lambda -1/2$,
\be \label{Iepsa}
\begin{aligned}
 I_\eps(a\eps^\lambda) &= \int_\R |
 \eta |^{2s}  \left( \e^{-a \left(1-\widehat M_{\lambda, \eps} (\eta)\right )}
-\e^{-a}
\right )^2\, \d \eta = \frac {\e^{-2a}}{\eps^{2s+1}} \int_\R |\xi|^{2s} \left( \e^{\frac a{1+|\xi|^\lambda}}-1\right)^2\, \d \xi \\
&= C_{s,a} \frac 1{\eps^{2s+1}} <+\infty,
\end{aligned}
\ee
 for $t\geq a \geq a\eps^\lambda$ we get
 \be\label{stima} I_\eps(t)= I_\eps(a\eps^\lambda)\e^{-{2\eps^{-\lambda}}(t-a\eps^\lambda)} + {2\eps^{-\lambda}}  \int_{a\eps^\lambda}^t
A_\eps(\sigma) \e^{-{2\eps^{-\lambda}}(t-\sigma)}\, \d \sigma.
 \ee
 Thanks to
estimate \eqref{A(t)} we get
 \be\label{stima2} \int_{a\eps^\lambda}^t  A_\eps(\sigma)
\e^{-{2\eps^{-\lambda}}(t-\sigma)}\, \d \sigma \leq  \frac {C}
{\eps^{2s+1}} \int_{a\eps^\lambda}^t \frac { \e^{-{2\eps^{-\lambda}}(t-\sigma)}}
{\sigma^{\beta(2s+1)}} \, \d \sigma.
 \ee
 Integrating by parts
\[
\begin{aligned}
\int_{a\eps^\lambda}^t
\frac { \e^{-{2\eps^{-\lambda}}(t-\sigma)}}
{\sigma^{\beta(2s+1)}}
\, \d \sigma
& = \frac {\eps^\lambda}2 \frac 1{t^{\beta(2s+1)}} -\frac {\eps^\lambda}2 \frac { \e^{-2  {\eps^{-\lambda}}(t-{a\eps^\lambda})}}{{a}^{\beta(2s+1)}\eps^{\lambda\beta(2s+1)}}
+ \frac{\eps^\lambda} 2 {\beta (2s+1)}
\int_{a\eps^\lambda}^t
\frac { \e^{-{2\eps^{-\lambda}}(t-\sigma)}}
{\sigma^{\beta(2s+1)+1}}
\, \d \sigma \\
&\leq
\frac {\eps^\lambda}2 \frac 1{t^{\beta(2s+1)}} +
\frac {\eps^\lambda}2\frac {\beta (2s+1)}  {a\eps^\lambda}
\int_{a\eps^\lambda}^t
\frac{ \e^{-{2\eps^{-\lambda}}(t-\sigma)}}
{\sigma^{\beta(2s+1)}}
\, \d \sigma.
\end{aligned}
\]
Let us choose $a$ such that $1-\frac {\beta (2s+1)}  {2a}  >0$ ($a$ depends on $s$ and
$\beta$). We obtain
\[
\int_{a\eps^\lambda}^t
\frac { \e^{-{2\eps^{-\lambda}}(t-\sigma)}}
{\sigma^{\beta(2s+1)}}
\, \d \sigma
\leq C_{s,\beta} \frac {\eps^\lambda}{t^{\beta(2s+1)}}.
\]
Getting back to \eqref{Iepsa}, \eqref{stima} and \eqref{stima2}, we proved that there
exists a constant  $ C=C(\lambda, s,\beta) >0$ such that
 \be\label{stima3} I_\eps(t) \leq
{C}\frac 1{ \eps^{2s+1}} \left( \e^{-{2t\eps^{-\lambda}}} +\frac 1
{t^{\beta(2s+1)}}\right ).
 \ee
 Finally, by \eqref{sob}, \eqref{stima4} and \eqref{stima3},
we get
\[
\left\| \tilde P_{\lambda, \eps, reg}( t) \right\|_{\dot H^s}^2\leq  C \frac {(1+t)^{{(2s+1)/\lambda}}}{\eps^{2s+1}} \left(
 \e^{-{2t\eps^{-\lambda}}} + \frac1 { t^{\beta(2s+1)}}\right ) \leq C\frac 1 {\eps^{2s+1}}(1+t)^{(2s+1)\left({1 /\lambda} -\beta\right )}
\]
and  denoting $C_2 = C$
\[
 \left\| \tilde P_{\lambda, \eps, reg}( t) \right\|_{\dot H^s}\leq  C_2 \frac 1 {\eps^{s+{1 /2}}}(1+t)^{(s+{1 /2})\left({1 /\lambda} -\beta\right )}.
\]

\section{Conclusions}
In this paper we studied an approximation to fractional diffusion equations obtained
by using the argument originally proposed for the linear diffusion equation by Rosenau
 \cite{Ros}. As it happens for the linear diffusion, the approximation coincides with
 a  linear kinetic equation of Boltzmann type, in which the Maxwellian background is
 now
represented by a Linnik distribution \cite{L, L2}. A detailed analysis of the solution
to this kinetic equation allows us to obtain various interesting properties. Among
others, it was interesting to discover that the solution to the Rosenau approximation
can be split into two parts, easily identified in terms of their regularity: one
singular, and the other regular. The former simply represents a perturbation of mass
zero, and it decays to zero exponentially both with respect to time and to the small
parameter $\ve$ characterizing the approximation. The latter is shown to approach in
time, for any fixed value of the parameter $\ve$, the fundamental solution to the
fractional diffusion equation in strong sense. This allows us to conclude that the
Rosenau argument introduces in a natural way a consistent approximation of fractional
diffusion equations, which not only reproduces the limit phenomenon at fixed time and
for small values of the parameter $\ve$, but also reproduces, apart of rapidly decaying
perturbations, the limit phenomenon for large times. Out of doubts, these results
could be fruitfully employed to construct new numerical approximations to  fractional
diffusion equations.

\vskip 2cm \noindent {\bf Acknowledgements.} This work has been written within the
activities of the National Group of Mathematical Physics (GNFM) and of the National Group of Mathematical Ana\-ly\-sis, Probability and Applications
(GNAMPA) of INDAM. Two authors
(AP) and (GT) acknowledge support  by MIUR project ``Optimal mass transportation,
geometrical and functional inequalities with applications''. We thank Federico
Bassetti for useful discussions about the domain of normal attraction of a stable law.


%
%


\end{document}